\documentclass[10pt,twoside]{article}
\usepackage[utf8]{inputenc}
\usepackage[T1]{fontenc}
\usepackage[english]{babel}
\usepackage{amsmath,amssymb,amsfonts,multido,amsthm}
\usepackage{pstricks}
\usepackage{titlesec}
\usepackage{indentfirst}

\usepackage{hyperref}
\usepackage{graphicx}
\usepackage{color}
\definecolor{matheonblue}{RGB}{0,0,50}
\definecolor{matheonlightblue}{RGB}{139,0,0}
\usepackage[inner=2.4cm,outer=2.65cm,bottom=3cm]{geometry}
\titleformat*{\section}{\normalsize\bfseries}
\titleformat*{\subsection}{\normalsize\itshape}
\DeclareMathOperator{\R}{\mathbb{R}}
\DeclareMathOperator{\Rc}{\mathcal{R}}
\DeclareMathOperator{\C}{\mathcal{C}}
\DeclareMathOperator{\S2}{\mathbb{S}^2}

\DeclareMathOperator{\ra}{\rightarrow}
\DeclareMathOperator{\de}{\text{d}}
\DeclareMathOperator{\tr}{tr}

\newcommand{\f}[1]{{\boldsymbol{ #1}}}

\newcommand{\pat}[2]{\frac{\partial #1}{\partial #2}}
\DeclareMathOperator{\di}{\nabla\cdot}

\newcommand{\curl}{\nabla \times}

\newcommand{\intet}[1]{\int_{\Omega}{ #1} \de \f x}
\newcommand{\sk}[1]{\left (\nabla \f {#1}\right )_{\mathrm{skw}}}
\newcommand{\dd}{\tilde{\f d}}
\newcommand{\skw}{\mathrm{skw}}
\newtheorem{proposition}{Proposition}[section]
\date{Version \today}
\begin{document}

\author{Etienne Emmrich\footnote{Technische Universit\"{a}t Berlin,
Institut f\"{u}r Mathematik, Stra{\ss}e des 17.\ Juni 136,
10623 Berlin, Germany. Email: {\texttt emmrich@math.tu-berlin.de}}
\and
Sabine H.~L.~Klapp\footnote{Technische Universit\"{a}t Berlin,
Institut f\"{u}r Theoretische Physik,
Hardenbergstra{\ss}e 36, 10623 Berlin, Germany. Email: {\texttt klapp@physik.tu-berlin.de}}
\and
Robert Lasarzik\footnote{Technische Universit\"{a}t Berlin,
Institut f\"{u}r Mathematik, Stra{\ss}e des 17.\ Juni 136,
10623 Berlin, Germany. Email: {\texttt lasarzik@math.tu-berlin.de}}
}

\title{%
\begin{Large}
\textbf{Nonstationary models for liquid crystals: A fresh mathematical perspective}
\end{Large}
}
\maketitle
\begin{abstract}
\normalsize
In this article we discuss nonstationary models for inhomogeneous liquid crystals driven out of equilibrium
by flow. Emphasis is put on
those models which are used in the mathematics as well as in the physics literature, the overall goal
being to illustrate the mathematical progress on popular models
which physicists often just solve numerically.
Our discussion includes the Doi--Hess model for the orientational distribution function,
the $Q$-tensor model and the Ericksen--Leslie model which focuses on the director dynamics.
We survey particularly the mathematical issues (such as existence of solutions) and linkages between these models. Moreover,
we introduce the new concept of relative energies measuring the distance between solutions
of equation systems with nonconvex energy functionals and discuss possible applications of this concept for future studies.
\end{abstract}

\section{Introduction}

Since their discovery in the 1890s bei Reinitzer~\cite{reinitzer} and Lehmann~\cite{lehmann} (see also Heinz~\cite{heinz2} and Virchow~\cite{virchow} for earlier descriptions), liquid crystals continue to be one
of the most intriguing and fascinating classes of condensed matter, which nowadays 
have a plethora of applications in optics, photonics, and in material science (for a review, see Ref.~\cite{peter}).

Typical liquid crystals consist of rod-like ("prolate") or disk-shaped ("oblate")
organic molecules or colloidal particles. The name ``liquid crystal'' already suggests their intermediate role between
two more conventional states of matter. These are, on the one hand,
fully isotropic liquids, which lack of any (positional or orientational) long-range ordering and form the common case for most high-temperature atomic and molecular
fluids. On the other hand, crystals, which represent the typical low-temperature state of many materials, are characterized by three-dimensional positional 
(and possibly orientational) order.
Between these cases, liquid crystals are characterized by long-range orientational
ordering of the axes of the anisotropic particles without (or with only partial)
ordering of the positions of the center of mass. As a result, liquid crystals can {\em flow}. More generally, liquid crystals respond easily
to external thermal, mechanical, or optical perturbations and are therefore typical representatives of {\em soft} condensed matter systems. 
The unique structural and dynamical material properties of liquid crystals continues to attract an interdisciplinary community of physicists, chemists, material
scientists, and even (applied) mathematicians. This interest is recently also triggered by the important role
of liquid-crystal physics in the fields of biophysics (\emph{e.g.}, for the
structure of the cytoskeleton or the movement between actin and myosin, see Ahmadi~\emph{et\,al}.~\cite{christina}),
in active matter Ref.~\cite{marchetti}
 and in astrophysics (emergence of topological defects). Many of these contexts 
involve physical situations outside thermal equilibrium, where the material properties generally depend on time.
The purpose of the present article is to summarize modeling
approaches for such {\em nonstationary} (out-of-equilibrium) liquid crystals from both, a mathematical and a physical
perspective.

Clearly, the presence of orientational degrees of freedom makes the theoretical description of liquid crystals 
more complex than that of ordinary (atomic) fluids. This holds for microscopic ("bottom-up") approaches
such as classical density functional theory (see Ref.~\cite{brader,hansen}), but also for coarse-grained approaches
such as phase-field crystal modeling (see Ref.~\cite{wittkowski})
and for mesoscopic (continuum) approaches involving appropriate
order parameter fields (see Ref.~\cite{hess2015}) or even macroscopic variables, such as a stress tensor~(see Sec.~\ref{sect-el}). Such mesoscopic approaches
have become particularly popular for the description of liquid crystals under {\em flow}, a situation of major relevance
for many applications (see Ref.~\cite{zhou}). Mathematically,  continuum approaches for nonstationary liquid crystals
involve typically nonlinear coupled 
(partial) differential equations. While physicists just tend to solve these equations numerically and explore the emerging physical behavior, there are many
open problems from the mathematical (and numerical) side concerning, \emph{e.g.}, the existence and uniqueness of solutions.
From the physical side, this poses the danger of overseeing important dynamical features, while  from the mathematical
side, there is a certain risk to concentrate on too simplistic (or even unphysical) models. 

It is in this spirit that we here aim at giving an overview over some of the most relevant
nonstationary models that have been considered in {\em both}, the physics and the mathematics literature,
and to elucidate their challenges. These challenges become particularly apparent when treating {\em inhomogeneous} systems under flow: here
one is faced not only with the impact of the flow field on the structure of the liquid crystal, but also
{\em vice versa} with the structure-induced modification on the flow.  
For a recent review focusing on homogeneous situations alone, see 
Ref.~\cite{zhou}.

We concentrate on theories targeting the {\em nematic} state of liquid crystals. Indeed, 
depending on the degree of order one distinguishes different mesophases, the two simplest ones of which 
are illustrated in the left part of Fig.~\ref{figure}. 
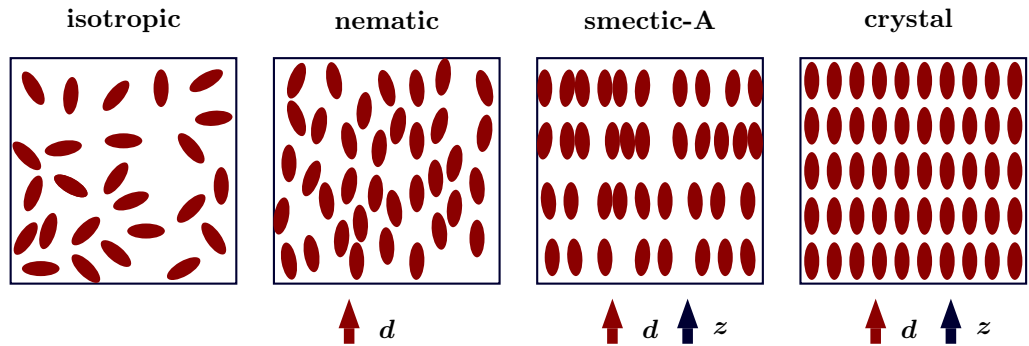
\begin{figure}[htb]
\begin{center}
\psset{xunit=1cm,yunit=1cm}
\begin{pspicture}(0,-0.4)(13.5,4)
\rput(1.5,3.5){\textbf{isotropic}}
\pspolygon[linecolor=matheonblue](0,0)(0,3)(3,3)(3,0)
\rput{90}(.4,.2){\psellipse*[linecolor=matheonlightblue](0,0)(.1,.25)}
\rput{45}(1,.2){\psellipse*[linecolor=matheonlightblue](0,0)(.1,.25)}
\rput{-33}(.2,.6){\psellipse*[linecolor=matheonlightblue](0,0)(.1,.25)}
\rput{163}(.5,.7){\psellipse*[linecolor=matheonlightblue](0,0)(.1,.25)}
\rput{50}(1.4,0.4){\psellipse*[linecolor=matheonlightblue](0,0)(.1,.25)}
\rput{-60}(2.3,0.2){\psellipse*[linecolor=matheonlightblue](0,0)(.1,.25)}
\rput{-45}(1,0.7){\psellipse*[linecolor=matheonlightblue](0,0)(.1,.25)}
\rput{-90}(1.8,0.7){\psellipse*[linecolor=matheonlightblue](0,0)(.1,.25)}
\rput{35}(2.7,0.6){\psellipse*[linecolor=matheonlightblue](0,0)(.1,.25)}
\rput{-20}(0.3,1.2){\psellipse*[linecolor=matheonlightblue](0,0)(.1,.25)}
\rput{60}(0.8,1.3){\psellipse*[linecolor=matheonlightblue](0,0)(.1,.25)}
\rput{45}(.215,1.7){\psellipse*[linecolor=matheonlightblue](0,0)(.1,.25)}
\rput{-80}(0.7,1.8){\psellipse*[linecolor=matheonlightblue](0,0)(.1,.25)}
\rput{-35}(1.4,1.4){\psellipse*[linecolor=matheonlightblue](0,0)(.1,.25)}
\rput{-70}(1.6,1.1){\psellipse*[linecolor=matheonlightblue](0,0)(.1,.25)}
\rput{-45}(2.4,1){\psellipse*[linecolor=matheonlightblue](0,0)(.1,.25)}
\rput{0}(2.8,1.3){\psellipse*[linecolor=matheonlightblue](0,0)(.1,.25)}
\rput{90}(1.5,1.9){\psellipse*[linecolor=matheonlightblue](0,0)(.1,.25)}
\rput{40}(2.4,1.8){\psellipse*[linecolor=matheonlightblue](0,0)(.1,.25)}
\rput{30}(0.3,2.6){\psellipse*[linecolor=matheonlightblue](0,0)(.1,.25)}
\rput{60}(0.8,1.3){\psellipse*[linecolor=matheonlightblue](0,0)(.1,.25)}
\rput{-5}(0.8,2.5){\psellipse*[linecolor=matheonlightblue](0,0)(.1,.25)}
\rput{-40}(1.4,2.5){\psellipse*[linecolor=matheonlightblue](0,0)(.1,.25)}
\rput{0}(2.0,2.6){\psellipse*[linecolor=matheonlightblue](0,0)(.1,.25)}
\rput{-60}(2.6,2.7){\psellipse*[linecolor=matheonlightblue](0,0)(.1,.25)}
\rput{-85}(2.7,2.2){\psellipse*[linecolor=matheonlightblue](0,0)(.1,.25)}

\rput(5,3.5){\begin{minipage}{3cm}\begin{center}
\textbf{{nematic}}
\end{center}
\end{minipage}}

\pspolygon[linecolor=matheonblue](3.5,0)(3.5,3.0)(6.5,3)(6.5,0)
\rput{10}(3.7,.3){\psellipse*[linecolor=matheonlightblue](0,0)(.1,.25)}
\rput{9}(4.0,.4){\psellipse*[linecolor=matheonlightblue](0,0)(.1,.25)}
\rput{-5}(4.4,.6){\psellipse*[linecolor=matheonlightblue](0,0)(.1,.25)}
\rput{0}(4.6,.3){\psellipse*[linecolor=matheonlightblue](0,0)(.1,.25)}
\rput{3}(5.0,.5){\psellipse*[linecolor=matheonlightblue](0,0)(.1,.25)}
\rput{-2}(5.4,.3){\psellipse*[linecolor=matheonlightblue](0,0)(.1,.25)}
\rput{7}(5.7,.7){\psellipse*[linecolor=matheonlightblue](0,0)(.1,.25)}
\rput{0}(6.2,.6){\psellipse*[linecolor=matheonlightblue](0,0)(.1,.25)}
\rput{-10}(3.6,.9){\psellipse*[linecolor=matheonlightblue](0,0)(.1,.25)}
\rput{6}(4.2,1.1){\psellipse*[linecolor=matheonlightblue](0,0)(.1,.25)}
\rput{3}(4.65,.8){\psellipse*[linecolor=matheonlightblue](0,0)(.1,.25)}
\rput{-4}(4.85,1.2){\psellipse*[linecolor=matheonlightblue](0,0)(.1,.25)}
\rput{9}(5.1,1.0){\psellipse*[linecolor=matheonlightblue](0,0)(.1,.25)}
\rput{0}(5.4,1.2){\psellipse*[linecolor=matheonlightblue](0,0)(.1,.25)}
\rput{2}(5.65,1.4){\psellipse*[linecolor=matheonlightblue](0,0)(.1,.25)}
\rput{-5}(5.9,1.0){\psellipse*[linecolor=matheonlightblue](0,0)(.1,.25)}
\rput{6}(6.2,1.3){\psellipse*[linecolor=matheonlightblue](0,0)(.1,.25)}
\rput{-8}(6.3,2.0){\psellipse*[linecolor=matheonlightblue](0,0)(.1,.25)}
\rput{0}(3.7,1.6){\psellipse*[linecolor=matheonlightblue](0,0)(.1,.25)}
\rput{-20}(4,1.4){\psellipse*[linecolor=matheonlightblue](0,0)(.1,.25)}
\rput{-10}(4.5,1.3){\psellipse*[linecolor=matheonlightblue](0,0)(.1,.25)}
\rput{20}(3.8,2.2){\psellipse*[linecolor=matheonlightblue](0,0)(.1,.25)}
\rput{10}(4.5,1.9){\psellipse*[linecolor=matheonlightblue](0,0)(.1,.25)}
\rput{-5}(4.9,1.8){\psellipse*[linecolor=matheonlightblue](0,0)(.1,.25)}
\rput{0}(5.4,1.9){\psellipse*[linecolor=matheonlightblue](0,0)(.1,.25)}
\rput{-10}(5.9,1.6){\psellipse*[linecolor=matheonlightblue](0,0)(.1,.25)}
\rput{-20}(3.8,2.7){\psellipse*[linecolor=matheonlightblue](0,0)(.1,.25)}
\rput{-10}(4.1,2.1){\psellipse*[linecolor=matheonlightblue](0,0)(.1,.25)}
\rput{10}(4.3,2.7){\psellipse*[linecolor=matheonlightblue](0,0)(.1,.25)}
\rput{-5}(4.7,2.3){\psellipse*[linecolor=matheonlightblue](0,0)(.1,.25)}
\rput{15}(5.0,2.6){\psellipse*[linecolor=matheonlightblue](0,0)(.1,.25)}
\rput{-10}(5.15,2.1){\psellipse*[linecolor=matheonlightblue](0,0)(.1,.25)}
\rput{5}(5.4,2.6){\psellipse*[linecolor=matheonlightblue](0,0)(.1,.25)}
\rput{-7}(5.75,2.75){\psellipse*[linecolor=matheonlightblue](0,0)(.1,.25)}
\rput{-15}(5.7,2.1){\psellipse*[linecolor=matheonlightblue](0,0)(.1,.25)}
\rput{15}(6.3,2.6){\psellipse*[linecolor=matheonlightblue](0,0)(.1,.25)}

\rput(8.5,3.5){\begin{minipage}{3cm}\begin{center}
\textbf{{smectic-A}}
\end{center}
\end{minipage}}
\pspolygon[linecolor=matheonblue](7,0)(7,3.0)(10,3)(10,0)
\rput{0}(7.2,.35){\psellipse*[linecolor=matheonlightblue](0,0)(.1,.25)}
\rput{4}(7.5,.35){\psellipse*[linecolor=matheonlightblue](0,0)(.1,.25)}
\rput{02}(7.9,.35){\psellipse*[linecolor=matheonlightblue](0,0)(.1,.25)}
\rput{-5}(8.4,.35){\psellipse*[linecolor=matheonlightblue](0,0)(.1,.25)}
\rput{0}(8.7,.35){\psellipse*[linecolor=matheonlightblue](0,0)(.1,.25)}
\rput{-3}(9.2,.35){\psellipse*[linecolor=matheonlightblue](0,0)(.1,.25)}
\rput{0}(9.5,.35){\psellipse*[linecolor=matheonlightblue](0,0)(.1,.25)}
\rput{05}(9.8,.35){\psellipse*[linecolor=matheonlightblue](0,0)(.1,.25)}

\rput{5}(7.15,1.1){\psellipse*[linecolor=matheonlightblue](0,0)(.1,.25)}
\rput{0}(7.45,1.1){\psellipse*[linecolor=matheonlightblue](0,0)(.1,.25)}
\rput{-3}(7.9,1.1){\psellipse*[linecolor=matheonlightblue](0,0)(.1,.25)}
\rput{0}(8.1,1.1){\psellipse*[linecolor=matheonlightblue](0,0)(.1,.25)}
\rput{06}(8.4,1.1){\psellipse*[linecolor=matheonlightblue](0,0)(.1,.25)}
\rput{0}(8.7,1.1){\psellipse*[linecolor=matheonlightblue](0,0)(.1,.25)}
\rput{03}(9.1,1.1){\psellipse*[linecolor=matheonlightblue](0,0)(.1,.25)}
\rput{-2}(9.4,1.1){\psellipse*[linecolor=matheonlightblue](0,0)(.1,.25)}
\rput{0}(9.8,1.1){\psellipse*[linecolor=matheonlightblue](0,0)(.1,.25)}

\rput{-7}(7.1,1.9){\psellipse*[linecolor=matheonlightblue](0,0)(.1,.25)}
\rput{0}(7.4,1.9){\psellipse*[linecolor=matheonlightblue](0,0)(.1,.25)}
\rput{3}(7.6,1.9){\psellipse*[linecolor=matheonlightblue](0,0)(.1,.25)}
\rput{0}(8.0,1.9){\psellipse*[linecolor=matheonlightblue](0,0)(.1,.25)}
\rput{4}(8.2,1.9){\psellipse*[linecolor=matheonlightblue](0,0)(.1,.25)}
\rput{0}(8.4,1.9){\psellipse*[linecolor=matheonlightblue](0,0)(.1,.25)}
\rput{5}(8.9,1.9){\psellipse*[linecolor=matheonlightblue](0,0)(.1,.25)}
\rput{-7}(9.2,1.9){\psellipse*[linecolor=matheonlightblue](0,0)(.1,.25)}
\rput{0}(9.45,1.9){\psellipse*[linecolor=matheonlightblue](0,0)(.1,.25)}
\rput{-3}(9.7,1.9){\psellipse*[linecolor=matheonlightblue](0,0)(.1,.25)}
\rput{4}(9.9,1.9){\psellipse*[linecolor=matheonlightblue](0,0)(.1,.25)}

\rput{0}(7.1,2.6){\psellipse*[linecolor=matheonlightblue](0,0)(.1,.25)}
\rput{-7}(7.4,2.6){\psellipse*[linecolor=matheonlightblue](0,0)(.1,.25)}
\rput{3}(7.6,2.60){\psellipse*[linecolor=matheonlightblue](0,0)(.1,.25)}
\rput{0}(7.9,2.6){\psellipse*[linecolor=matheonlightblue](0,0)(.1,.25)}
\rput{0}(8.1,2.6){\psellipse*[linecolor=matheonlightblue](0,0)(.1,.25)}
\rput{-5}(8.4,2.60){\psellipse*[linecolor=matheonlightblue](0,0)(.1,.25)}
\rput{0}(8.9,2.6){\psellipse*[linecolor=matheonlightblue](0,0)(.1,.25)}
\rput{03}(9.2,2.6){\psellipse*[linecolor=matheonlightblue](0,0)(.1,.25)}
\rput{-4}(9.6,2.60){\psellipse*[linecolor=matheonlightblue](0,0)(.1,.25)}
\rput{03}(9.9,2.60){\psellipse*[linecolor=matheonlightblue](0,0)(.1,.25)}


\rput(12,3.5){%
\begin{minipage}{3cm}%
\begin{center}%
{\textbf{crystal}}
\end{center}
\end{minipage}
}
\pspolygon[linecolor=matheonblue](10.5,0)(10.5,3)(13.5,3)(13.5,0)
\multido{\rr=0.3+0.6}{5}{
\multido{\i=0+1,\r=10.65+0.3}{10}{%
\rput{0}(\r,\rr){\psellipse*[linecolor=matheonlightblue](0,0)(.1,.25)}%
}%
}

\psline[linewidth=4pt,arrowinset=0.0,arrowsize=8pt,linecolor=matheonlightblue]{->}(4.5,-0.8)(4.5,-0.2)
\psline[linewidth=4pt,arrowinset=0.0,arrowsize=8pt,linecolor=matheonlightblue]{->}(8,-0.8)(8,-0.2)
\psline[linewidth=4pt,arrowinset=0.0,arrowsize=8pt,linecolor=matheonblue]{->}(9,-0.8)(9,-0.2)
\psline[linewidth=4pt,arrowinset=0.0,arrowsize=8pt,linecolor=matheonlightblue]{->}(11.5,-0.8)(11.5,-.2)
\psline[linewidth=4pt,arrowinset=0.0,arrowsize=8pt,linecolor=matheonblue]{->}(12.5,-0.8)(12.5,-0.2)
\rput(5,-0.6){
$\f d$
}
\rput(8.5,-0.6){
$\f d$
}
\rput(12,-0.6){
$\f d$
}
\rput(9.5,-0.6){
$\f z$
}
\rput(13,-0.6){
$\f z$
}

\end{pspicture}
\end{center}
\caption{Schematic illustrating the molecular structure in the most common phases of liquid crystals. The director is denoted by $\f d$ and the layer normal by $\f z$. \label{figure}}
\end{figure}
The isotropic state is characterized by translational and orientational disorder.  Upon increasing the concentration (lyotropic liquid crystal) or lowering temperature (thermotropic liquid crystal) the systems then enters the nematic state, see, \emph{e.g.}, Palffy-Muhoray~\cite{peter}.
Within the equilibrium nematic state, the rod-like molecules are randomly
distributed in space but tend to align along a common direction,
characterized by the so-called director $\f d$. The resulting system is much more permeable for light sent along 
the direction of the vector $\f d$ than perpendicular
to it. This is the fundamental principle underlying all liquid-crystals displays (LCDs). 
We note that in nonequilibrium, (induced, \emph{e.g.}, by flow) the symmetry of the originally nematic phase can break down in the sense that the orientational order becomes inhomogeneous or biaxial.  

In equilibrium smectic phases, 
the orientational alignment is supplemented by long-range one-dimensional
translational order, leading to layers of aligned particles.
The corresponding material density exhibits periodic peaks in one
direction (see vector $\f z$ in Fig.~\ref{figure}).
Both, nematic and smectic phases phases have relevance in material science, such as, \emph{e.g.}, for liquid
crystal elastomers which deform as a reaction to thermal, chemical or optical
excitement (see Ref.~\cite{peter}). It is clear, however, that the
one-dimensional translational order in the smectic phase poses additional challenges
for the theoretical description. We thus focus on systems which are nematic in equilibrium.

The remainder of this article is structured as follows. We start by 
shortly reviewing the Doi model (Sec.~\ref{sect-doi}),
which provides a "kinetic" equation for the dynamical evolution
of the orientational distribution function of the system.
We then proceed to the $Q$-tensor (Sec.~\ref{sect-q}) and, as a special case (valid for uniaxial nematics),
the Ericksen--Leslie model (Sec.~\ref{sect-el}), both of which focus on the dynamics
of an order parameter. For all of these model,
we consider both homogeneous and inhomogeneous situations, and
we discuss their advantages and drawbacks from the mathematical and physical
point of view. Special emphasis is put on the existence of generalized solutions
 of the governing equations.  
 In Sec.~\ref{sect-rel}, we then
comment on cross-linkages between the three models. 
In Sec.~\ref{sect-rela}, we then introduce a new piece of mathematical analysis for nonstationary models of liquid crystals. Specifically, we introduce the concept of relative energies to measure the distance of possible solutions. We prove an inequality showing the continuous dependence of the relative energy on the difference of the initial values.
The article closes
with a brief outlook in Sec.~\ref{sect-open}.

\paragraph{Notation} For two vectors $\f a, \f b \in \mathbb{R}^3$, we denote by
$\f a\cdot \f b = \f a^T \f b$ the Euclidean inner product, by $\f a \times \f b$ the vector or outer
product, and by $\f a\otimes \f b = \f a \f b^T$ the dyadic product.
The identity in $\mathbb{R}^{3\times 3}$ is denoted by $I$. For $A, B \in \mathbb{R}^{3\times 3}$, let
$A:B = \sum_{i,j=1}^3 A_{ij} B_{ij}$ be the Frobenius inner product.
For a vector field $\f v$, we denote by $(\nabla \f v)_{\mathrm{sym}}= \tfrac{1}{2}\left(\nabla \f v +
(\nabla \f v) ^T\right)$ the symmetric and
by $(\nabla \f v)_{\mathrm{skw}}= \tfrac{1}{2}\left(\nabla \f v - (\nabla \f v) ^T\right)$
the skew-symmetric part of its gradient $\nabla \f v$. With $\curl$ we denote the curl operator.
By $\Omega$ we denote a bounded domain with sufficiently smooth boundary.

\section{The Doi model}\label{sect-doi}

From a microscopic point of view, the liquid-crystalline systems can be considered
as an ensemble of (rigid or flexible) rods which perform thermal motion and interact pairwise via repulsive 
(''steric") and, possibly, additional attractive (\emph{e.g.}, van-der Waals) forces. These pairwise forces can then be supplemented
by coupling to a flow field or other external perturbations.
Within the Doi model (see Refs.~\cite{kuzuu,doi}), which is sometimes
also referred to as Doi--Hess model
(see Beris and Edwards~\cite[p.~463]{beris}) due to the parallel work of Hess (see Refs.~\cite{Hess} and~\cite{hess2015}),
the fluctuating ensemble of rods is described by a (nonnegative) probability distribution function 
$f=f(\f x, \f n  , t)$ where $f$ is the probability density 
that at time $t$ a molecule at point $\f x$ is aligned in
direction $\f n $. The latter is an
element of the unit sphere $\S2$ in $\R^3$ such that $|\f n | = 1$.
Here and henceforth we assume the particles themselves to have uniaxial symmetry, such
that one vector is fully sufficient to describe their orientation.

The key element of the Doi--Hess approach is a Fokker--Planck
equation for the probability density $f=f(\f x, \f n  , t)$ describing its temporal (or spatio-temporal) evolution. 
Mathematically, these are (partial) differential equations of first order in time.
 Note that the first-order nature implies
that the motion of the molecules is assumed to be ''overdamped'', a situation which typically occurs if the rods are suspended
in a viscous solvent. Due to this assumption one often rather refers to a Smoluchowski (see Ref.~\cite{dhont}) than a Fokker--Planck equation.

Before discussing these dynamical equations, we first briefly state two important properties
of $f$: First, $f$ is normalized such that
\begin{equation}\label{normed}
\int_{\mathbb{S}^2}f(\f x, \f n , t) \de \f n  = 1
\end{equation}
for all $\f x \in \Omega$ and all $t\ge 0$, where $\Omega \subset \mathbb{R}^3$
denotes the spatial domain occupied by the liquid crystal.
For a fully isotropic and homogeneous fluid, where the directions and positions are uniformly distributed and $f$ is just a constant,
Eq.~(\ref{normed}) immediately implies that $f=1/(4\pi)$.

A second important property of the distribution function $f$ concerns its symmetry:
Because of the head-to-tail symmetry characterizing many (nonpolar) rod-like molecules, one 
typically cannot
distinguish between the directions $\f n $ and $-\f n $. It would thus be convenient to
consider the manifold $\R\!P^2$, the real projective plane in  $\R^3$, instead of
$\mathbb{S}^2$. Since this manifold is, however, nonorientable,
one rather endows the function $f$ with the  symmetry assumption
\begin{equation}
f(\f x, \f n  ,t) = f ( \f x ,- \f n ,t)
\end{equation}
for all $(\f x, \f n  ,t) \in \Omega \times \S2 \times [0,\infty)$.
%
%
\subsection{Homogeneous systems: Doi--Onsager model}\label{sect-doi-hom}

In the absence of any spatial structure, that is, for a spatially
homogeneous system, the probability density $f$ reduces 
to a function of the orientation (and time) such that
$f= f ( \f n  ,t )$. In this case, the Doi model is often referred to as the Doi--Onsager model (see, \emph{e.g.}, Zhang and Zhang~\cite{zhangreview}).

The main orientational phenomenon for a homogeneous liquid-crystalline system in {\em thermal equilibrium}
is the phase transition between the isotropic state (where $f=1/(4\pi)$)
and the nematic state, where $f$ is a nontrivial function of the
orientation (yet stationary in time). In his pioneering work of 1949, 
Onsager~\cite{onsager} proposed the first statistical-physical theory
of the equilibrium isotropic-nematic transition in a system of
rods with purely repulsive interactions. Within this theory,
the transition results from a competition between the "excluded volume" 
interactions (originating from the mutual repulsion of the rods) and the orientational
entropy. 
We will come back to some elements of the Onsager theory later in this section.

Here we proceed by first stating, in a quite general form, the governing equation
describing the temporal evolution of the probability density $f$ (often called "kinetic" equation) in the homogeneous case.
In writing this equation, we allow for an orientation-dependent potential $V=V( \f n , f( \f n  ,t))$ affecting
the orientation of the rods, and for a flow velocity field $\f v$ such that $\nabla \f v$ is constant.
The governing equation then reads 
\begin{equation}\label{doi}
\pat{f(\f n , t)}{t}= \frac{1}{\text{De}}\,\mathcal{R} \cdot
\Big( \mathcal{R} f(\f n , t) + f(\f n , t)\mathcal{R} V(\f n  , f (\f n , t))\Big)
 - \mathcal{R}\cdot \Big( \f n  \times \nabla \f v \f n  f(\f n , t) \Big) \, .
\end{equation}
Here,  $\Rc$ is a rotational differential operator to be specified below, and
the nonnegative constant $\text{De}$ is the so-called Deborah number, which
quantifies the ratio between the rotational Brownian motion
of the molecules and the motion due to convection (note that Eq.~(\ref{doi}) is already in dimensionless form,
see Yu and Zhang~\cite{yuzhang} for the scaling arguments).

The first term on the right-hand side of Eq.~(\ref{doi}), ${\text{De}}^{-1} \Rc \cdot \Rc f$,
describes the rotational diffusion of the rods due to rotational Brownian motion.
The operator $\Rc$ depends on $\f n $ and is defined as
\begin{equation}
\Rc  = \f n  \times \pat{}{\f n } \, ,
\end{equation}
it can be seen as the gradient with respected to $\f n $ restricted to the
sphere. The composition $\Rc\cdot \Rc$ then represents the Laplace--Beltrami
operator on the sphere $\mathbb{S}^2$ (see also Zhang and Zhang~\cite{zhangreview}).
With $|\f n | = 1$, we find
\begin{align}
\begin{split}
\Rc \cdot \Rc f&= \sum_{j=1}^3 \left( 1 - n_j^2 \right) \frac{\partial^2 f}{\partial n_j^2}
- \sum_{i , j =1, i\neq j}^3 n_i n_j \frac{\partial^2 f}{\partial n_i \partial n_j}
- 2 \sum_{j=1}^3 n_j \frac{\partial f}{\partial n_j}
\\
& = \Delta_{\f n } f - (\f n  \otimes \f n ) : \nabla_{\f n }^2 f - 2 \f n  \cdot \frac{\partial f}{\partial \f n }
\,  ,
\end{split}
\end{align}
where $\Delta_{\f n } $ is the Laplacian and $\nabla_{\f n }^2$ is the Hessian with respect to $\f n $. In spherical coordinates we have
$\f n  = [\cos \phi \sin \theta , \sin \phi \sin \theta, \cos \theta]^T$ with $\phi \in [0,2\pi)$ and $\theta \in [0,\pi)$. The diffusion term
then becomes
\begin{equation}
\Rc \cdot \Rc f = \frac{1}{\sin^2 \theta}\, \frac{\partial^2 f}{\partial \phi^2} + \frac{\partial^2 f}{\partial \theta^2} + \cot \theta\,
\frac{\partial f}{\partial \theta} \, .
\end{equation}

We now consider the second term on the right-hand side of Eq.~(\ref{doi})
which involves the potential $V$ describing effects of alignment.
This potential is, in general, a function of the direction $\f n $ and the probability density $f$,
\emph{i.e.}, $V = V(\f n , f(\f n , t))$.
Many different
choices of $V$ have been suggested in the literature. 
Here we follow Kuzuu and Doi~\cite{kuzuu} 
who write the potential as a sum of two parts,
$V = V_{MI} + V_H$, where $V_{MI}$ models 
pair interactions on a mean-field level,
and $V_H$ accounts for the interaction of each rod with an external magnetic field (if present). 

The latter term is a single-particle contribution and thus depends only on $\f n $.
It can be written as (see also Ref.~\cite[Sec.~3.2.1]{gennes}),
\begin{equation}\label{VH}
V_H ( \f n  ) =  - \frac{1}{2} (\chi_\|-\chi_\bot) ( \f n  \cdot \f H )^2 \, ,
\end{equation}
where $\f H$ is the externally controlled magnetic field and $\chi_\|$ and $ \chi_\bot $ are
the (magnetic) susceptibilities parallel to the molecular axis $\f n $
and perpendicular to $\f n $, respectively.

The mean-field potential $V_{MI} $ describes the (effective) interaction of the particles and
is given by
\begin{equation}
V_{MI}( \f n , f) =
\int_{\mathbb{S}^2} \beta( \f n , \boldsymbol{\hat{n}}) f ( \boldsymbol{\hat{n}} ,t) \de
\boldsymbol{\hat{n}}
\, ,\label{VMI}
\end{equation}
where the integral kernel $\beta=\beta( \f n , \boldsymbol{\hat{n}})$
describes the interaction of two molecules pointing in the directions $\f n  $ and
$\boldsymbol{\hat{n}}$. Again, different choices of $\beta$
can be found in the literature. In the original Onsager model~\cite{onsager},
the kernel $\beta$ is given by
\begin{equation}\label{onsager}
\beta( \f n , \boldsymbol{\hat{n}}) = \alpha\, |\f n  \times \boldsymbol{\hat{n}}|
\end{equation}
with a positive constant
$\alpha$. The latter represents the coupling strength (see Refs.~\cite{gennes,kuzuu}) and depends
on quantities characterizing the microscopic configuration. Another  kernel function that is often studied is due to Maier and
Saupe (see Ref.~\cite{MeierSaupe}),
\begin{align}\label{MaierSaupe}
\beta ( \f n , \boldsymbol{\hat{n}} )
\sim -  (\f n  \cdot \boldsymbol{\hat{n}} )^2 = | \f n  \times \f{\hat{n}}|^2 -1 \, .
\end{align}
This expression can be considered as an approximation of the cross product (for small angle) in the corresponding Onsager term
(\ref{onsager}). A third example is the "dipole-like" potential given by
\begin{align}\label{dipol}
\beta ( \f n , \boldsymbol{\hat{n}}) \sim -  \f n  \cdot \boldsymbol{\hat{n}} \, .
\end{align}
We stress that, from a physical perspective, Eq.~(\ref{dipol}) corresponds to a mean-field
version of the full dipole-dipole potential, which involves an additional dependency
on the connecting vector between the particles (see Ref.~\cite{jackson}). Still, Eq.~(\ref{dipol})
reflects one key feature of the interaction between molecules are possessing a dipole moment (see Fatkullin and Slastikov~\cite{fatkullin}):
That is, the interaction potential changes its sign if the direction of one of the molecules is reversed.
This is in contrast to the Maier--Saupe and Onsager potential. 

The last term on the right-hand side in Eq.~\eqref{doi}
models the impact of an imposed flow profile onto the
alignment of the molecules. Note that since $f$ was assumed to be independent of $\f x$
(homogeneity),  the velocity gradient
$\nabla \f v$ needs to be constant. Explicitly, the cross product
yields
\begin{equation}
\f n  \times \nabla \f v  \f n  = \sum_{j=1}^3
\begin{bmatrix}\displaystyle
 n_2  \, \pat{v_{3}}{x_j}  -  d_3 \, \pat{v_2}{x_j}\\[0.5ex]
 \displaystyle
 n_3  \, \pat{v_{1}}{x_j} -  d_1\, \pat{v_{3}}{x_j}\\[0.5ex]
 \displaystyle
 n_1 \, \pat{v_{2}}{x_j} -  d_2 \, \pat{v_{1}}{x_j}
\end{bmatrix}
n_j.
\end{equation}
Combined with the application of the rotational operator one obtains (with $|\f n |=1$)
\begin{align}
\begin{split}
\Rc \cdot \Big(\f n  \times \nabla \f v \f n  f\Big) &=
\left (\f n  \times \pat{}{\f n } \right ) \cdot
\left ( \f n  \times \nabla \f v \f n  f \right ) \\
&= \left ( \f n  \cdot \f n  \right ) \left ( \pat{}{\f n }
\cdot \nabla \f v \f n  f \right )  -
  \pat{}{\f n } \cdot \f n   \left ( \f n  \cdot \nabla \f v \f n  f \right ) \\
 & = \pat{}{\f n }\cdot \big(\left ( | \f n |^2 I - \f n  \otimes \f n  \right ) \nabla \f v \f n  f \big)
 \\
 & = (\nabla \cdot \f v) f - 3 \f n  \cdot \nabla \f v \f n   f
+   (\nabla \f v \f n )\cdot   \left ( I - \f n  \otimes \f n  \right ) \pat{}{\f n }f
 \, .
 \end{split}
\end{align}
Having discussed the physical interpretation of the different terms in Eq.~(\ref{doi}),
we close with some more mathematical remarks.
Importantly, the evolution equation (\ref{doi})
can be interpreted as the {\em gradient flow} of a certain (free) energy functional
with respect to the Wasserstein metric (see Villani~\cite{opttrans}). 
To be precise, consider the free energy functional 
\begin{equation}\label{free}
F[f]= 
  \int_{\mathbb{S}^2} \Big( f(\f n , t)
\left( \ln f(\f n ,t) - 1 \right)
+ \frac{1}{2} V_{MI}( \f n ,f( \f n ,t)) f( \f n  ,t)
+V_H ( \f n  ) f ( \f n ) \Big )\de \f n 
  \, ,
\end{equation}
where the notation $F[f]$ indicates the functional dependence of $F$ on the probability density.
Performing the variational derivative of $F$ with respect to $f$
one obtains (using the definition of $V_{MI}$ in Eq.~\eqref{VMI} and the symmetry of $\beta$)
\begin{align}
\frac{\delta F	}{\delta f}[f] = \ln f + V_{MI}+ V_H=\ln f + V  \, .
\end{align}
Here $\delta F /\delta f$ denotes the variational derivative of $F$ with respect to $f$ (see Ref.~\cite{furihata}).
Consider now the first two terms on the right-hand side of Eq.~(\ref{doi}) or, equivalently, the evolution equation in
the absence of (or for a constant) velocity field $\f v$. 
By using the functional derivative given above
it is easy to see that
\begin{align} 
\label{doions_2}
\begin{split}
\frac{\partial f}{\partial t} & = 
\frac{1}{\text{De}}\Rc \cdot \Big( \Rc f + f \Rc V \Big)\\
&= \frac{1}{\text{De}}\Rc \cdot \left ( f  \Rc \left ( \frac{\delta F}{\delta f }[f]\right )\right ) \, .
\end{split}
\end{align}
The first term on the right-hand side of Eq.~\eqref{doi} with the interaction potential $V$ 
can thus be interpreted as the gradient flow of the free energy functional $F$~\eqref{free} with respect to the Wasserstein metric.
In other words, the free energy functional is a Lyapunov-type function for Eq.~\eqref{doi}.

We note in passing that the first member of Eq.~(\ref{doions_2}) is also referred to as Smoluchowski
equation on the sphere (see Ref.~\cite{existencedoi}).
Global existence of a solution to Eq.~\eqref{doions_2} with the Maier--Saupe
potential~\eqref{MaierSaupe}
has been shown in Ref.~\cite{existencedoi}  for initial data given by
a nonnegative continuous function on $\S2$.
Moreover, it has been shown that solutions are nonnegative and normalized.
Furthermore, energy estimates and the typical structure of steady-state solutions
have been studied.  
Many other authors have also studied Eq.~\eqref{doions_2} from the mathematical point of view
(for a review, see Zhang and Zhang~\cite{zhangreview}). 
In particular, Zhang and Zhang~\cite{zhangreview} consider the three
above-mentioned interaction kernels \eqref{onsager}, \eqref{MaierSaupe}, and
\eqref{dipol}. Other results concerning the homogeneous case can be found
in Refs.~\cite{fatkullin,Constantin,zhangsymmetry}.

Finally, it is interesting to note that even the full evolution equation (\ref{doi}) for homogeneous systems,
which involves a coupling to a flow field with spatially constant $\nabla \f  v$,
can be written as gradient flow, as long as $\nabla \f v $ is symmetric, \emph{i.e.}, $\nabla \f v = ( \nabla \f v )_{\text{sym}}$. This can be achieved by adding to the free energy functional (\ref{free})
the ''dissipation'' potential 
\begin{equation}
W[ f] =\frac{1	}{2}  \int_{\mathbb{S}^2}( \f n  \cdot (\nabla \f  v)_{\text{sym}} \f n  ) f \de \f n  \, .
\end{equation}

A typical case with constant derivative of the
velocity field is a system in a planar Couette shear flow.
Kuzuu and Doi~\cite{kuzuu} derived from this model the homogeneous Ericksen--Leslie equation (see Sec.~\ref{sect-el}) by taking into account perturbations
of the equilibrium state at a
small Deborah number, $\text{De}$. 
%
\subsection{Inhomogeneous Doi model}

Inhomogeneous flow situations occur, \emph{e.g.}, as a result of an externally imposed flow field with velocity gradient depending on the
spatial variable $\f x$.
Another intriguing possibility is a {\em spontaneous} breaking of the symmetry, resulting
in the formation of shear bands (see Ref.~\cite{cates}).

For inhomogeneous systems, the probability density $f$
depends on $\f x$ (as well as on orientation $\f n $ and time $t$), and the dynamical evolution equation \eqref{doi} is no longer valid.
The required modifications of Eq.~(\ref{doi}) concern, first, the 
incorporation of translational diffusion (and other translational effects, see Ref.~\cite[Sec.8]{doi}).
To be specific, the governing kinetic equation for $f$ in the inhomogenous case is given by
Ref.~\cite{existinstatdoi}
\begin{align}
\begin{split}
\pat{f}{t} + ( \f v \cdot \nabla ) f = &
\frac{\varepsilon^2}{\text{De}}\di \Big(\left (D_\|^* \f n  \otimes \f n  + D_\bot^*
\left ( I- \f n  \otimes \f n  \right )\right )
 \left ( \nabla f + f\nabla V\right )\Big)
\\&+\frac{1}{\text{De}} \mathcal{R} \cdot \left ( \mathcal{R} f +f\mathcal{R} V\right )
- \mathcal{R}\cdot ( \f n  \times \nabla \f v \f n   f ) \, .
\end{split}
\label{instatdoi}
\end{align}
In the first term on the right-hand side, the parameter $\varepsilon$ quantifies the ratio between the lengths of the rods and a characteristic length
describing the spatial extension of the flow region. Further, $D_\|^*$ and $D_\bot^*$ are shape dependent, translational diffusion constants
characterizing the diffusion parallel and perpendicular to the molecular axis $\f n $, respectively.
Typically one has $D_\|^* > D_\bot^*$ (see Ref.~\cite{dhont}), that is, a molecule moves more easily parallel (than perpendicular) to $\f n $, consistent with the
na\"ive perception.
Equation~(\ref{instatdoi}) also allows for a spatial dependence 
of the alignment potential $V=V(\f x, \f n , f(\f x, \f n  , t))$ (see Sec.~\ref{sect-doi-hom} where we introduced this quantity for the homogeneous case),
which may occur, \emph{e.g.}, through an inhomogeneous magnetic field, or through a 
spatial dependence of the kernel of the mean-field interaction $V_{MI}$.
Usually one assumes this kernel (now called  $\widetilde{\beta}$)
to be translationally invariant Ref.~\cite{zhang}, which
implies that $\widetilde{\beta}$ depends on $\f x - \boldsymbol{\hat{x}}$ instead of
$\f x$ and $\boldsymbol{\hat{x}}$ separately. A typical ansatz
reads Ref.~\cite{existinstatdoi} 
\begin{align}
\widetilde{\beta}( \f x - \hat{\f x}, \f n  , \hat{\f n  })
:= \, \chi\left (\f x - \hat{\f x}\right ) \,
\beta ( \f n  ,\hat{\f n  }),
\end{align}
where $\chi$ denotes a suitable mollifier\footnote{Consider the compactly supported smooth function $\rho\in\C_c^\infty(\R^d)$, defined by $ \rho(x) = \begin{cases} c \exp\left (-\frac{1}{1-| x|^2}\right ) & \mbox{if}~ |x|<1 \\
0 & \mbox{else}
\end{cases}\,.$
The constant $c$ is chosen such that $ \int_{\R^d} \rho(x)\de x = 1$. The suitable mollifiers $\chi $ are then defined by $ \chi(x) =(1/L^3) \rho(x / L)$. The constant $L$ describes the interaction radius of the molecules in the material. 
} 
modeling the range of interaction of the rods, and
the orientation-dependent function $\beta$ appearing on the right-hand side is given, \emph{e.g.}, by one of the ansatzes \eqref{onsager},~\eqref{MaierSaupe}, or~\eqref{dipol}.
The resulting mean-field potential
$V_{MI}$ then becomes
\begin{align}
V_{MI}(\f x, \f n , f)
= \int_{\Omega} \int_{\mathbb{S}^2}
\widetilde{\beta}( \f x - \boldsymbol{\hat{x}} , \f n  , \boldsymbol{\hat{n}})
f ( \boldsymbol{\hat{x}} , \boldsymbol{\hat{n}} ,t) \de \boldsymbol{\hat{n}}
\de \boldsymbol{\hat{x}} \, .
\end{align}

The second main modification (as compared to the kinetic description of homogeneous flow)
is that the evolution equation for $f=f(\f x,\f n ,t)$ has to be supplemented
by an equation of motion for the velocity field $\f v$.
The latter can be derived from the conservation of momentum and the incompressibility condition resulting
from the conservation of mass. This results in a 
Navier--Stokes-like system of equations given by
\begin{align}
\begin{split}
\pat{\f v}{t} + ( \f v \cdot \nabla )\f v + \nabla p &= \di \sigma + \f b\, ,\\
\di \f v &=0 \, .
\end{split}
\label{navdoi}
\end{align}
Here $p$ denotes the isotropic (scalar) pressure (up to an additive constant), 
$\sigma$ is the (additional) stress tensor describing frictional effects (see below),
and $\f b$ is a body force per unit mass.

We note that, for highly viscous fluids, the first member of Eq.~\eqref{navdoi} is often replaced by its limiting form valid for low Reynolds numbers (overdamped limit). In this limit, the entire material derivative, that is, the time derivative and the nonlinear term in the velocity, are set to zero, yielding the so-called stationary Stokes equation.

Irrespective of the choice for the material derivative,
the stress tensor $\sigma$ of a complex fluid like a liquid crystal is commonly written as
 a sum of an
elastic and a viscous part, \emph{i.e.},
\begin{align}
\sigma=\sigma_{\text{elast}} + \sigma_{\text{visc}},
\end{align}
where the viscous part $\sigma_{\text{visc}}$ is given by (see Ref.~\cite[Sec.8.6]{doi} for a derivation)
\begin{align}\label{sigmav}
\sigma_{\text{visc}} = 2 \nu (\nabla \f v)_{\mathrm{sym}} + \frac{1}{2}\, \xi_r \int_{\mathbb{S}^2}
\left( (\nabla \f v)_{\mathrm{sym}} : \f n  \otimes \f n \right ) \f n  \otimes \f n  f( \cdot ,\f n , \cdot) \de \f n \, ,
\end{align}
with $\nu$ and $\xi_r$ being the viscosity of the fluid and the rotational friction constant, respectively.
The dots in the integral indicate the dependence of $f$ on $\f x$ and $t$, which are also transferred to the tensors $\sigma_{\text{visc}}$ and $(\nabla \f v)_{\mathrm{sym}}$.
The first term on the right-hand side of Eq.~\eqref{sigmav} 
describes the isotropic friction; it already occurs in an atomic fluid and thus also appears in the
standard Navier--Stokes equation.
Note that $2\nabla\cdot (\nabla \f v)_{\mathrm{sym}} = \Delta \f v$ for an incompressible fluid.
The second term on the right-hand side of Eq.~\eqref{sigmav} represents the
additional friction due to the motion of the rod-like molecules.
Depending on the physical state considered, the viscous stress tensor can still be a function
of $\f x$ and $t$.

For the elastic part of the stress tensor, Wang {\em et\,al}. \cite{wangzhang} have suggested the expression
\begin{equation}
\sigma_{\text{elast}} = - \int_{\mathbb{S}^2} \left (\f n  \times \left (\Rc \frac{\delta F}{\delta f}[f]\right )
\right )\otimes \f n  f (\cdot, \f n  \cdot) \de \f n ,\label{sigmaelast}
\end{equation}
and for the body force appearing in Eq.~\eqref{navdoi}
\begin{equation}
\quad \f b = \int_{\mathbb{S}^2} \nabla\frac{\delta F}{\delta f }[f]  \, f( \cdot , \f n  ,\cdot ) \de \f n  \, .
\end{equation}

The system of equations for the inhomogeneous flowing liquid crystals given by Eqs.~\eqref{instatdoi} and \eqref{navdoi} has
been studied in a number of mathematically oriented publications.
Specifically, Zhang and Zhang~\cite{existinstatdoi} have shown local existence as
well as global existence for small Deborah number
and large viscosity.
Numerical simulations for some special cases such as
 plane Couette and Poiseuille flow have been presented in Ref.~\cite{yuzhang}.

In an earlier study,  Doi~\cite{earlydoi} considered a simplified case of Eq.~\eqref{instatdoi} where 
the molecular interactions are neglected, \emph{i.e.},the potential $V$ is assumed to be constant, and
 the translational
diffusion is isotropic (\emph{i.e.}, $D_\|^* = D_\bot ^*$). 
Eq.~\eqref{instatdoi} then reduces to (see also Ref.~\cite{otto})
\begin{align}
\pat{f}{t}+ ( \f v\cdot\nabla) f=& - \pat{}{\f n } \cdot
\left ((\nabla \f v \f n  - ( \f n  \cdot \nabla \f v \f n ) \f n   ) f \right )
+ D \Delta f + \frac{1}{\text{De}} \Rc \cdot \Rc f\, ,
\label{fokker}
\end{align}
where $D= D_\|^*\varepsilon^2 / \text{De}$.
This equation~\eqref{fokker} represents again an overdamped Fokker--Plank equation (see, \emph{e.g.}, Ref.~\cite{fokker}). 
If $V$ is assumed to be constant, the elastic part of the stress tensor~\eqref{sigmaelast} could be calculated using integration by parts:
\begin{align}
\sigma_{\text{elast}} ={}& - \int_{\mathbb{S}^2} \left ( \f n  \times \Rc f \otimes \f n  \right ) \de \f n  
 =- \int_{\mathbb{S}^2} f \pat{}{\f n  } \cdot  \left ( \f n  \otimes  (| \f n  |^2 I-\f n  \otimes \f n   ) \right ) \de \f n  
= 3 \int_{\mathbb{S}^2} \left ( \f n  \otimes \f   n - \frac{1}{3} I  \right ) f \de \f n  \,.\label{sigmaelastsim}
\end{align}
More general expressions for the stress tensors can be found
in Refs.~\cite{doi} and~\cite{fokker}.
The definition~\eqref{sigmaelastsim} is very similar to the definition of the $Q$-tensor in Eq.~\eqref{Q} below.
For the body force $\f b$ we infer, since $f$ is a probability measure (see Eq.~\eqref{normed}), that
 \begin{align}
 \f b = \int_{\mathbb{S}^2} \nabla  f(\f x  , \f n  ,t ) \de \f n = \nabla  \int_{\mathbb{S}^2}   f( \f x , \f n  , t ) \de \f n = 0 \, .
 \end{align}
 for all $\f x \in \Omega$ and $t>0$. 
Finally, the rotational friction is assumed to vanish (compare to Eq.~\eqref{sigmav}), \emph{i.e.}, $\xi_r=0$. 
As in Ref.~\cite{Bae1}, one finally ends up with the system consisting of Eq.~\eqref{fokker} and the Navier--Stokes-like equation,
\begin{align}
\begin{split}
\pat{\f v}{t} + (\f v \cdot \nabla )\f v + \nabla p =& \nu\Delta\f v +\di\sigma_{\text{elast}} \, ,
\\
\di \f v =&0 \, ,
\end{split}\label{fokkernav}
\end{align}
where $\sigma_{\text{elast}}$ is given  in Eq.~\eqref{sigmaelastsim}.
 Otto and Tazavaras~\cite{otto} have been able
to show existence of strong solutions to the initial and boundary value problem for Eq.~\eqref{fokkernav} in the limit of large viscosity $\nu$, where the Navier--Stokes-like equation
reduces to the (stationary) Stokes equation with an additional stress tensor.
 Bae and Trivisa~\cite{Bae1} have
proved global existence of weak solutions to Eqs.~\eqref{fokker} and~\eqref{fokkernav}
and extended
this result also to compressible fluid flow~(see~Bae and Trivisa~\cite{Bae2}).
For more general stress tensors $\sigma $, certain regularity results are
derived in Ref.~\cite{fokker,fokkerregu}.


\section{$Q$-tensor theory}\label{sect-q}
\subsection{The $Q$-tensor and corresponding free energy functionals\label{sect-qstation}}

So far we have focused on the dynamics of the entire space- and orientation dependent probability density,
$f=f( \f x, \f n ,t)$. A common simplification,
pioneered by Nobel price laureate Pierre G.\ de Gennes (see de Gennes~\cite{gennes}),
consists in studying rather the dynamical evolution of the
{\em lowest}-order moment of $f$. For a liquid crystalline comprised of {\em nonpolar} particles,
the lowest (nonvanishing) moment is the so-called $Q$-tensor, a second-rank tensorial
quantity which provides a complete description of the orientational state of the system. The dynamics
is then referred to as $Q$-tensor dynamics.

At this point it seems worth to state the major {\em physical} argument 
why one should generally use a second-rank order parameter tensor to describe 
the orientational structure and dynamics rather than just the average orientation,
which would correspond to the nematic director. The latter is a vectorial
quantity and thus seems, at least on first sight, easier to handle (indeed, the dynamics
of the nematic director within the so-called Ericksen--Leslie theory  will be discussed 
in Sec.~\ref{sect-el}). However, reducing the system's dynamics to that of the director implicitly
assumes that the orientational order is {\em uniaxial}. This assumption can
break down in the presence of topological defects within the nematic 
phase (see Ref.~\cite{sonnet,hydro}), but also
as a result of strong shear flow (see Ref.~\cite{rienacker})
inducing complex states of orientational motion. 

To introduce $Q$, we recall (see Sec.~\ref{sect-doi}) that $f$ 
is a normalized, positive definite probability density on the sphere which is invariant against
reversal of the molecular orientation $\f n $. Thus, 
the first moment is zero, that is 
\begin{align}
\int_{\mathbb{S}^2} \f n  f (\f x , \f n  , t ) \de \f n =0\quad\text{for all }\f x \in \Omega, t \geq 0\,.
\end{align}
The second moment $Q$-tensor then is given as (see, \emph{e.g.}, de Gennes~\cite[Sec.~2.1]{gennes})
\begin{equation}\label{Q}
Q( \f x, t):= \int_{\mathbb{S}^2} \left (\f n  \otimes \f n  - \frac{1}{3}\, I\right )
f ( \f x, \f n ,t) \de \f n  \, .
\end{equation}
It is obvious that $Q$ is symmetric and thus possesses a complete system of orthonormal eigenvectors
$\f e_i$ with real eigenvalues $\lambda_i$  ($i= 1,2,3$). 
Therefore, $Q$ can be represented by the spectral decomposition (see Ref.~\cite{majumdar}),
\begin{align}\label{spectral}
Q = \lambda_1 \f e_1\otimes \f e _1 + \lambda_2 \f e _2\otimes \f e _2
+ \lambda_3 \f e _3\otimes \f e_3 \,.
\end{align}
Since the eigenvectors $\f e_i$ are unit vectors and since $f$ is a probability density, 
the eigenvalues are bounded (see Ref.~\cite{mark}) according to 
\begin{align}\label{eigen}
-\frac{1}{3} \leq \lambda_i =\f e_i\cdot Q \f e_i
= \int_{\S2}( \f e_i\cdot\f n  )^2 f (\f x,\f n ,t ) \de \f n  - \frac{1}{3}\leq \frac{2}{3}\, .
\end{align}
This condition is often called physical condition (see Ref.~\cite{mark}).
Further, because of $|\f n | = 1$, the tensor $Q$ is always traceless such that
\begin{equation}
\label{Qtrace}
\tr Q =  \lambda_1 + \lambda_2 + \lambda_3 = 0 \, .
\end{equation}

An analysis of the eigenvalues yields information about
the ordering state of the system. Specifically, $Q$ allows to distinguish between three different states:
isotropic, uniaxial or biaxial. In the isotropic phase, $f=1/(4\pi)$.
A short calculation then shows that
the tensor $Q$ is zero since $\int_{\mathbb{S}^2} \f n \otimes \f n \de \f n = (4\pi / 3) I $, and so are the eigenvalues $\lambda_i$.
Uniaxial nematic phases are characterized 
by a spectrum where two of the eigenvalues of $Q$ coincide, whereas the third one 
is different. Finally, biaxial states are characterized by three different eigenvalues
(see Mottram and Newton~\cite{intro}). 

We briefly mention an alternative way to decompose the $Q$-tensor (see
Ref.~\cite{majumdar}), which requires
only two eigenvectors  but two additional scalar order parameters.
This form is given by
\begin{align}
Q= s \left ( \f e_1 \otimes\f e_1 -\frac{1}{3}\,I\right )
+ r\left ( \f e_2\otimes \f e _2 - \frac{1}{3}\,I\right ).
\end{align}
The parameters $s$ and $r$ are related to the eigenvalues in Eq.~\eqref{spectral} via
\begin{align}
s = \lambda_1- \lambda_3= 2 \lambda_1+\lambda_2 , \quad r = \lambda_2-\lambda_3= 2\lambda_2+\lambda_1\,,
\end{align}
where we have used that $I = \f e_1 \otimes \f e_1 + \f e_2 \otimes \f e_2 + \f e_3 \otimes \f e_3 $.
Without loss of generality, $r=0$ in the uniaxial state (since $\lambda_2 = \lambda_3$ Ref.~\cite{majumdar2}),
but $r$ is nonzero in the biaxial state.

According to the standard thermodynamic principles, the equilibrium value of $Q$ corresponds 
to a minimum of a free energy functional. We thus need an expression for the free energy directly expressed
in terms of $Q$ (rather than in terms of the distribution $f$ as discussed in Sec.~\ref{sect-doi-hom}). 
As a starting point, we consider the famous Landau--de Gennes free energy (see Ref.~\cite{ball,HuangDing}), which
in the absence of external aligning fields and surfaces is given as
 \begin{subequations}
 \label{landau}
\begin{align}
F_{LG}[Q]=& 
\int_{\Omega} \Big (\frac{a(T)}{2}\tr\left (Q^2\right ) - \frac{b}{3}\, \tr(Q^3)
+ \frac{c}{4}\left (\tr \left (Q^2\right )\right )^2
\label{landau2} \\
&+\frac{ L_1 }{2}| \nabla Q|^2 + \frac{L_2}{2} \sum_{i=1}^3 \tr\left (\nabla Q_i \nabla Q_i\right )
+ \frac{L_3 }{2}|\di Q|^2\label{landau4}\\
& + \frac{L_4}{2} \sum_{i=1}^3 Q_i \cdot \curl Q_i+ \frac{L_5}{2} \sum_{i=1}^3 \tr\left (\nabla Q_i Q  (\nabla Q_i)^T\right )\Big )\de \f x.\label{landau3}
\end{align}
\end{subequations}
The first line on the right-hand side of Eq.~\eqref{landau2} describes the bulk free energy density,
assuming that the latter can be written as a polynomial  
in the $Q$ tensor. The coefficients
$b$ and $c$ are usually considered as state-independent material constants, whereas $a$ depends on the temperature
or composition (for a thermotropic or lyotropic system, respectively). Specifically, a change of the
sign of $a$ induces the isotropic-nematic transition. 
Note that all of the terms appearing in the bulk free energy are constructed to be rotationally invariant, as required
for a physically meaningful (that is, scalar) free energy $F$. This
requirement also enters into the remaining terms in Eq.~({\ref{landau}) which 
contain gradient terms and thus describe the ''elastic part'' of the free energy, that is,
the free energy due to elastic distortions.
Here, the notation $Q_i$ is used to describe the $i$-th row of
the $Q$-tensor (regarded as a column vector), and the parameters $L_1, \dots ,L_5$ are
material constants. A frequent assumption is the
one-constant approximation, where $L_2=L_3=L_4=L_5=0$~(see~Ref.~\cite{majumdar}). %

At this point it seems natural to ask about the relationship between the free energy functionals
for the $Q$ tensor, such as the Landau--de Gennes ansatz defined in Eq.~(\ref{landau}), and the ''microscopic" free energy functionals
for the full probability distribution considered in Sec.~\ref{sect-doi}. 
Indeed, both type of functionals are assumed to yield the {\em same} equilibrium state (described by $f$ or $Q$) by minimization!
A systematic strategy
to {\em derive} a $Q$-dependent free energy is as follows: One first expresses the full distribution $f$ as a power series in terms of
orientational order parameters
(involving $Q$ and higher-order moments).
Inserting and expanding the logarithmic (entropic) and interaction parts of the free energy up to the lowest vanishing terms
(which involve some closure approximation for the moments), one then obtains an expression for the bulk free energy
(for a recent application of this strategy, see Ref.~\cite{luga}%
).
The elastic contribution can then be derived by performing an additional gradient expansion.

Here we restrict ourselves to illustrating the relation between the $Q$- and the $f$-dependent functional
at the example of the 
Maier--Saupe interaction functional. The latter is given in Eq.~\eqref{MaierSaupe};
it represents the interaction part of the full $f$-dependent functional $F$ in Eq.~\eqref{free}.
Neglecting any spatial and temporal dependence of $f$, we find that
\begin{align}
\begin{split}
-&\int_{\S2}\int_{\S2} ( \f n \cdot \boldsymbol{\hat{n}})^2 f(\f n  ) f
(\boldsymbol{\hat{n}}) \de \f n  \de \boldsymbol{\hat{n}}=- \left (\int_{\S2}
\f n  \otimes \f n   f(\f n  )  \de \f n \right ):\left ( \int_{\S2} \boldsymbol{\hat{n}}\otimes
{\boldsymbol{\hat{n}}} f (\boldsymbol{\hat{n}})\de \boldsymbol{\hat{n}}\right )\\
&= - \left (\int_{\S2}  \left (\f n  \otimes \f n  -\frac{1}{3}I \right )f(\f n  )
\de \f n \right ):\left ( \int_{\S2} \left (\boldsymbol{\hat{n}}\otimes
{\boldsymbol{\hat{n}}}- \frac{1}{3}I\right ) f (\boldsymbol{\hat{n}})\de \boldsymbol{\hat{n}}\right )\\
&\hspace{1em}+\frac{1}{3} \int_{\S2} \hat{\f n }\cdot \hat{\f n }f ( \hat{\f n })
\de \hat{\f n }+\frac{1}{3} \int_{\S2}  {\f n }\cdot  {\f n }f (  {\f n }) \de  {\f n } -
 \frac{1}{9}= - \left |Q \right |^2 + \frac{5}{9}\,.
 \end{split}
\end{align}
The free energy then takes the form
\begin{align}\label{Maier}
F_{MS}[f]\sim  \int_{\S2} f(\f n ) (\ln f (\f n  )-1) \de \f n   -\alpha |Q|^2+ \frac{5}{9}\alpha\, ,
\end{align}
where we have not yet touched the entropic part (first term).  The additive constants $ \int_{\S2}f(\f n ) \de \f n  $ and
${(5\alpha)}/{9}  $ (with $\alpha$ being a constant related to the coupling strength, see also Eq.~\eqref{onsager})
can be neglected since the physical (equilibrium) behavior is determined by the derivative of $F$ rather
than by $F$ itself. The important point of Eq.~(\ref{Maier}) is that the (Maier--Saupe) interaction contribution is
simply expressed as the square of the norm of the order parameter. Historically, this was one of
the motivating ideas to construct a $Q$-dependent free energy.

At the end of this section on the stationary $Q$-tensor theory, we point the reader to
a quite omnipresent problem of the Landau--de Gennes free energy functional: 
Minimization of $F_{LG}$ (see Eq.~\eqref{landau}) can yield solutions for $Q$, whose
eigenvalues $\lambda_i$ {\em violate} the constraints given in Eq.~\eqref{eigen} (see Mottram and Newton~\cite{intro}).
To circumvent this problem,
Ball and Majumdar~\cite{ball} proposed the following strategy, which is
based on the free energy $F_{MS}$ given in Eq.~\eqref{Maier}.
The idea is to minimize $F_{MS}$ 
over all probability densities $f$ for a fixed
$Q$-tensor (which has the required properties). 
To this end, one first defines the set of probability densities related to a given $Q$-tensor via
\begin{align}
\mathcal{A}_Q:= \left \{  f:\S2 \ra \R, f\geq 0, \int_{\S2}f(\f n )\de \f n =1\,;
\quad Q= \int_{\S2}\left ( \f n  \otimes \f n  - \frac{1}{3}\,I\right )f(\f n ) \de \f n 
\right \}.\label{A}
\end{align}
For this set $\mathcal{A}_Q$, the authors of Ref.~\cite{ball} considered the minimization
problem associated to the energy~\eqref{Maier} and defined the function
\begin{align}
g(Q):=\left \{ \begin{array}{ll}
\min_{f \in \mathcal{A}_Q} \int_{\S2} f(\f n ) \ln f (\f n  ) \de \f n  \,
& \text{if }-\frac{1}{3}\leq\lambda_i(Q)\leq \frac{2}{3}\quad \text{for all }i=1,2,3\,,\\
\infty\,& \text{otherwise}.
\end{array}\right .
\label{mini}
\end{align}
The function $g$ is finite if the eigenvalues of the given
$Q$-tensor fulfill the constraints in Eq.~\eqref{eigen} (see Ref.~\cite{ball}). 
One then obtains a new bulk free energy density 
(replacing that in Eq.~\eqref{landau2}), which is defined as
\begin{align}\label{bulk}
\psi_B(Q)= T g(Q) - \alpha |Q|^2 \, ,
\end{align}
where $T$ denotes the absolute temperature and $\alpha$ the coupling strength (see~Eq.~\eqref{Maier}).
Ball and Majumdar~\cite{ball} have investigated in detail the analytical
properties of $\psi_B(Q)$, like smoothness,
convexity, isotropy, boundedness from below and logarithmic blow-ups at the boundary
of the domain~(see also Ref.~\cite{mark}). This 
paves the way to usage of this functional in stationary and nonstationary problems.
Moreover, the functional can be easily extended toward inhomogeneous situations. 
Specifically, in the one-constant approximation one obtains 
\begin{align}\label{BM}
F_{BM}[Q]:=\intet{\left (L_1|\nabla Q|^2 + \psi_B(Q)\right )}\, .
\end{align}

An alternative strategy to "cure" the problem of obtaining eigenvalues of $Q$ beyond a prescribed range has been suggested 
by Heidenreich {\em et\,al}.~\cite{ilg}.
They proposed an "amended" free energy functional which coincides with the Landau--de Gennes ansatz for small
degree of ordering, but includes a correction term becoming effective for stronger nematic order. The correction
can be motivated by an expansion of the Onsager excluded-volume potential.
The amended free energy potential is given by~(see Ref.~\cite{ilg})
\begin{align}
F[Q]= \int_\Omega \left (  \frac{a(T)}{2 }| Q|^2 - \frac{b}{3} \tr(Q^3) - \frac{c^4_{\text{max}}}{2} \ln \left (  1- \frac{| Q|^4}{c^4_{\text{max}}} \right )    \right ) \de \f x \,.\label{amended}
\end{align}
The potential induces logarithmic blow-ups when the sum of the squares of the eigenvalues  of $Q$ approaches a certain threshold $c_{\text{max}}$, \emph{i.e.}~ $ (\sum_{i=1}^3 \lambda_i^2) \ra  c_{\text{max}}^2$. 
Choosing the norm in the definition~\eqref{amended} differently, taking rather the spectral than the Frobenius norm in the logarithmic term, would guarantee blow ups for $Q$-tensors with eigenvalues leaving the physical range~\eqref{eigen}.
The associated functional is given by
\begin{align}
F[Q]= \int_\Omega \left (  \frac{a(T)}{2 }| Q|^2 - \frac{b}{3} \tr(Q^3) - \frac{1}{8} \ln \left (  1- 4 \left \| Q+ \frac{1}{6}I\right \|_{2}^4 \right )    \right ) \de \f x \,,
\end{align}
where the spectral norm  $\|A\|^2_{2}$ of a symmetric matrix $A\in \R^{3\times 3}$ is 
just the largest absolute value of the eigenvalues.
This formulation leads to minimizers fulfilling the physical condition~\eqref{eigen}. 
In this way we avoid the infinite-dimensional minimization problem in~\eqref{mini}, which involves a special  set of function (see Eq.~\eqref{A}). 

\subsection{nonstationary equations}

We now turn to
$Q$-tensor theory for {\em flowing}, and possibly inhomogeneous, liquid crystals. 
To this end, various formulations of the equations have been proposed and studied in the literature. 
The first complete formulation is due to Hess~\cite{hess1975,hess1976},
who derived the equations using concepts from irreversible thermodynamics (see de Groot
and Mazur~\cite{de1984}
 and showed
that they can also be motivated via a Fokker--Planck approach (for a modern derivation, see~Hess~\cite{hess2015}).
 A closely related derivation was proposed by Olmstedt and Goldbart~\cite{olmsted,olmsted2},
 and by
Beris and Edwards \cite{beris}. We further mention the equations
proposed by Stark and Lubensky~\cite{stark}
derived from the Poisson bracket formalism, and the phenomenologically motivated
$Q$-tensor equations by Pleiner, Liu, and Brand~\cite{pleiner}.

All of these formulations agree in their general structure 
(comprising an evolution equation for $Q$ combined 
with the Navier--Stokes equation for incompressible fluids),
but differ in aspects such as the free-energy functional used for the relaxation term, 
and the types of coupling between the order parameter tensor and
the flow profile $\f v$.
 
In the following we focus on the "hydrodynamical" equations suggested by Edwards and Beris~\cite{beris},
which are widely used in the modern literature
as well as in many numerical studies (see, \emph{e.g.}, the Lattice--Boltzmann studies of Yeomans {\em et\,al}. \cite{hydro}).
These equations are given by
\begin{subequations}
\label{evo}
\begin{align}
\pat{Q}{t} + ( \f v \cdot \nabla ) Q - S( Q, \nabla \f v ) =& \Gamma  H \, ,\label{evoQ}\\
\pat{\f v }{t} + ( \f v \cdot \nabla )\f v + \nabla p
- \nu \Delta \f v =& \di ( \tau + \sigma) \, , \label{evov}\\
\di \f v =& 0 \, ,
\end{align}
\end{subequations}
where $Q=Q(\f x,t)$. In Eq.~(\ref{evoQ}), the first two terms on the left-hand side
form the material derivative of the $Q$-tensor (including
the advection stemming from the flow velocity field $\f v$).
The second term $S$ describes the influence of $\f v$ on
the spatio-temporal distribution of the order parameter.
More specifically, the (tensorial) velocity gradient $\nabla \f v$
can lead to stretching of $Q$, which is imposed by the symmetric part   $(\nabla \f v)_{\mathrm{sym}}$
of the velocity gradient
and to a rotation of $Q$, imposed by
the skew-symmetric part $(\nabla \f v)_{\mathrm{skw}}$.
Explicitly, the function $S$ is given by
\begin{align}
\label{S}
S(Q, \nabla \f v) := &\left( (\nabla \f v)_{\mathrm{skw}} + \xi (\nabla \f v)_{\mathrm{sym}}\right)\left (Q+\frac{1}{3}\,I\right )
-\left (Q+\frac{1}{3}\,I\right )( (\nabla \f v)_{\mathrm{skw}} - \xi (\nabla \f v)_{\mathrm{sym}})\notag \\& - 2 \xi
\left (Q+\frac{1}{3}\,I\right ) \tr (Q\nabla \f v )\, ,
\end{align}
where the parameter $\xi$ depends on the microscopic properties of the 
material considered. Taken altogether, 
the left-hand side of Eq.~\eqref{evoQ} is  similar to the {\em Oldroyd}
 derivative of the stress tensor appearing in the description of Oldroyd fluids
 (see Ref.~\cite{LionsMasmoudi}). 
 The Oldroyd system describes visoelastic materials by coupling the Navier--Stokes equation with an additional evolution equation for the stress tensor. 
In this latter equation, the material derivative is consistent with the left-hand side of Eq.~\eqref{evoQ} if $Q$ is replaced by the stress and the second line of Eq.~\eqref{S} is omitted (compare Ref.~\cite{LionsMasmoudi}).

Finally, the right-hand side of Eq.~(\ref{evoQ}) describes
the relaxation of $Q$ towards its equilibrium value
determined by the minimum of the free energy functional.
It thus involves (besides the rotational diffusion constant $\Gamma$)
the variational derivative of $F$ with respect to the order parameter, that is,
\begin{align}\label{H}
 H =- \frac{\delta F}{\delta Q}[Q] + \frac{1}{3}\,I\tr \frac{\delta F}{\delta Q}[Q] \, .
\end{align}
The precise form of the function $H$ obviously depends on the choice of the free energy functional
$F$. Common choices in the mathematical literature are the Landau de Gennes functionals $F_{LG}$ (see Eq.~\eqref{landau}) and 
its modified (constrained) version $F_{BM}$ (see Eq.~\eqref{BM}). From the physical side, most studies concentrate
on using $F_{LG}$.

The second member of Eq.~(\ref{evo}) describes the interplay between the velocity profile
and the orientational ordering. As in the case of the corresponding kinetic equations based
on the distribution function $f$ (see, \emph{e.g.}, Eq.~(\ref{navdoi})),
the coupling is achieved by adding to the standard Navier--Stokes
additional stress tensor contributions (note that the isotropic contributions
appear already on the left-hand side of Eq.~(\ref{evo})).  The $Q$-dependent 
contribution to the stress tensor consists of a symmetric part
\begin{align}
\label{stress_tau}
\tau= - \xi \left ( Q + \frac{1}{3}\,I\right ) H - \xi H \left ( Q+ \frac{1}{3}\,I\right )
+ 2 \xi \left ( Q + \frac{1}{3}\,I\right )  \tr \left ( Q H\right ) - \nabla Q : \pat{F}{\nabla Q}
\end{align}
and an skew-symmetric part
\begin{align}
\sigma= QH - HQ \, ,
\label{stress}
\end{align}
where $H$ is defined as in Eq.~\eqref{H}.
In Eq.~\eqref{stress_tau}, the last term depends on which terms are 
included in the elastic part (involving gradient terms) of the free energy,
see text below Eq.~(\ref{landau}). In case of the
one-constant approximation, we find that
\begin{align}
\left (\nabla Q : \pat{F}{\nabla Q}\right )_{ij} := \left (\nabla Q :
\pat{(L_1|\nabla Q|^2 )}{\nabla Q}\right )_{ij} = L_1\sum_{k,l=1}^3
\partial _{x_i} Q_{kl}\partial _{x_j} Q_{kl}.
\end{align}
What remains to be calculated is the quantity $H$, \emph{i.e.}, the functional derivative of the free energy with respect to $Q$ itself (see Eq.~\eqref{H}). 
Taking $F=F_{LG}$ (see Eq.~(\ref{landau})) and using again the one-constant approximation, we obtain 
\begin{align}\label{Hlg}
H =- a(T) Q + b\left (Q^2 -\frac{1}{3}\,I \tr \left (Q^2\right )\right )
- c Q \tr \left (Q^2\right )+ L_1 \Delta Q \, ,
\end{align}
where the first three terms on the right-hand side in Eq.~\eqref{Hlg} 
stem from the bulk part of the free energy, whereas the 
last term is due to the elastic part. 

Numerical (Lattice--Boltzmann) simulations of the present (Beris--Edwards) $Q$-tensor model have been studied, among others, by
Yeomans \emph{et\,al}.~\cite{hydro,yeomans,yeomans2} who focused on the emergence of topological defects
and on domain growth. 
The $Q$-tensor equations proposed by Hess have been numerically explored in detail in
their homogeneous version (see, \emph{e.g.}~Refs.~\cite{goetz,strehober2},
 who concentrated on
the complex time-dependent states appearing at large shear rates), in a simplified inhomogeneous version
without back-coupling to the velocity profile (see Ref.~\cite{das})
and in the fully coupled inhomogeneous version (see Ref.~\cite{mandal}).

We close this section with some remarks from the mathematical side.
The system given by  Eqs.~\eqref{evo}-\eqref{Hlg}, has been studied by Paicu and Zarnescu~\cite{existence2,existence} who
discussed well-posedness.  In Ref.~\cite{existence2},  they were able
to prove the global existence of weak solutions for the Cauchy problem in dimension
three. For the special case that the rotational parameter $\xi$ appearing, \emph{e.g.}, in Eq.~(\ref{S}) vanishes, they also demonstrated a certain higher regularity of the solution, 
as well as weak-strong uniqueness in dimension
two. In Ref.~\cite{existence}, the results are extended to non-zero $\xi$.

The Cauchy problem has also been  addressed in Huang and Ding~\cite{HuangDing} who considered the more general free energy~\eqref{landau} with specific assumptions on the
constants $L_i$ ($ i= 1,\dots,5$).
Additionally, the authors provided a well-posedness result for small data. 
Abels, Dolzmann and Liu~\cite{Abels} studied the Dirichlet and the
Neumann problem for the coupled Navier--Stokes/$Q$-tensor system
involving the Landau--de Gennes free energy~\eqref{landau} in the one-constant approximation
(similar to Eq.~\eqref{BM}). They proved global existence of weak solutions as well as local existence of
strong solutions together with regularity results.

Further, Wilkinson~\cite{mark} provided existence results for the case that the free energy
functional is chosen to have the form $F_{BM} $ (see Eq.~\eqref{BM}), again with the one-constant
approximation. Then $H$ takes the form
\begin{align}\label{Hbm}
H:= L_1 \Delta Q - T \left ( \pat{g}{Q}(Q)
- \frac{1}{3} \tr \left ( \pat{g}{Q}(Q)\right ) I\right ) + \alpha Q \,.
\end{align}
The resulting system consists  of Eqs.~\eqref{evo}-\eqref{stress} and Eq.~\eqref{Hbm}. For this case,
Wilkinson showed in Ref.~\cite{mark} global existence of  weak
solutions in three dimensions for periodic boundary conditions and additional regularity for
solutions in two dimensions.
We may emphasize 
 that the theory in Ref.~\cite{mark}
predicts $Q$-tensors as solutions fulfilling the physical constraint~\eqref{eigen}(compare to Sec.~\ref{sect-qstation}).

Finally, Feireisl \emph{et\,al}.~\cite{zarne}~considered a nonisothermal variant of system~\eqref{evo} with an extra
equation for the temperature and the solution fulfilling an entropy inequality. Existence of weak solutions is proved as well.


\section{Ericksen--Leslie theory}\label{sect-el}

As a special case of the $Q$ tensor theory outlined in Sec.~\ref{sect-q}, we now
consider {\em uniaxial} nematic states in nonstationary (and possibly inhomogeneous) situations. For an uniaxial (cylindrically symmetric) system, the $Q$-tensor takes the form
\begin{align}\label{Qd}
Q= s \left ( \f d \otimes \f d -\frac{1}{3}\,I\right ),
\end{align}
where  $\f d$ is the (normalized) eigenvector related to the eigenvalue with the largest absolute value.
This eigenvector is unique up to multiplication with $-1$; it
has algebraic multiplicity one due to the zero-trace condition (\ref{Qtrace}) of the $Q$-tensor.
Physically speaking, Eq.~(\ref{Qd}) expresses the fact that under the assumption of uniaxiality
(which typically holds only for liquid crystals at low molecular weights, see Beris and Edwards~\cite{beris}), 
the system's anisotropy is fully specified by the director $\f d$. 

In the 1960s, Ericksen~\cite{Erick1,Erick2} and Leslie~\cite{leslie,leslie2,leslie3} have proposed 
a set of nonstationary equations describing the director dynamics under shear. This "Ericksen--Leslie" theory
is nowadays regarded as a pioneer of the more advanced $Q$-tensor approach.
Besides the assumption of uniaxiality, a further key simplification within the Ericksen--Leslie approach 
consists of the fact is that the degree of order along $\f d$ is assumed to be {\em constant}, specifically, $|\f d|=1$ (similar to the
theory of Oseen~\cite{oseen} and Frank~\cite{frank} for inhomogeneous, stationary nematic states).
The resulting set of equations reads, in its most general form (see~Ref.~\cite{leslie}),
\begin{subequations}
\label{ELeq}
\begin{align}
\rho_1 \frac{\de^2 \f d}{\de t^2}  - \di \left ( \pat{F}{\nabla \f d} \right )+
\pat{F}{\f d}  - \lambda_1 \left ( \pat{\f d} {t}+ ( \f v \cdot \nabla ) \f d +
(\nabla \f v)_{\mathrm{skw}} ^T\f d \right ) + \lambda_2 (\nabla \f v)_{\mathrm{sym}} \f d &= \rho_1 \f G,\label{EL1}\\
\pat{\f v}{t} + ( \f v \cdot \nabla ) \f v + \nabla p + \di \left ( \nabla \f d ^T
\pat{F}{\nabla \f d }  \right ) - \di \sigma & = \f F, \label{EL2}\\
\di \f v &=0\,,\\
|\f d|^2 &=1\,.
\end{align}
\end{subequations}
The first term in Eq.~\eqref{EL1} is the second material derivative, where the first one is defined via $\de / \de t = \partial / \partial t + ( \f v \cdot \nabla) $, and $F$ is the free energy functional $F=F(\f d, \nabla \f d)$.
The vectorial quantities $\f G$ and $\f F$ appearing on the right-hand sides 
ob Eqs.~(\ref{EL1}) and (\ref{EL2}) represent external forces
acting on the director and on the velocity field, respectively. 
Equation~(\ref{EL2}) further involves the so-called "Leslie stress" $\sigma$ that corresponds
to the dissipative part of the stress tensor. It is given by
\begin{align}
\begin{split}
\sigma:= &\mu_1 (\f d \cdot (\nabla \f v)_{\mathrm{sym}} \f d )\f d\otimes \f  d + \mu_4 (\nabla \f v)_{\mathrm{sym}} + \mu_2 \f d \otimes \left (
\pat{\f d} {t}+ ( \f v \cdot \nabla ) \f d + (\nabla \f v)_{\mathrm{skw}} \f d \right )  \\&+ \mu_3
 \left ( \pat{\f d} {t}+ ( \f v \cdot \nabla ) \f d
+ (\nabla \f v)_{\mathrm{skw}} \f d \right )\otimes  \f d   + \mu_5  \f d \otimes \f d (\nabla \f v)_{\mathrm{sym}}+
\mu_6 (\nabla \f v)_{\mathrm{sym}}\f d \otimes \f d \,,
\end{split}\label{stressE}
\end{align}
where the "Leslie constants" $\mu_1,\ldots,\mu_4$ are related to those appearing in Eq.~(\ref{EL1}) via
(see Ref.~\cite{leslie}),
\begin{align}\label{parodi}
\lambda_1 = \mu_2-\mu_3, \qquad \lambda_2 = \mu_5-\mu_6= -(\mu_2+\mu_3)\,.
\end{align}
The last equation is often referred to as Parodi's relation. It is a consequence of the Onsager reciprocal relations (see~Ref.~\cite{parodi})

Regarding the free energy $F$, different choices have been considered. Leslie~\cite{leslie} suggested
to use the expression due to Oseen and Frank~\cite{frank}, that is,
\begin{align}\label{oseen}
F_{OF}(\f d,\nabla \f d ) := 
 k_1 (\di \f d )^2 + k_2 ( \f d \cdot \curl \f d +q )^2 + k_3 | \f d
\times \curl \f d |^2 + \alpha \left ( \tr (\nabla \f  d )^2 - (\di \f d )^2\right ),
\end{align}
where $k_i$ ($i=1,2,3$), $\alpha$ and $q$ denote elastic constants.
Note that we used the notation $F=F(\f d , \nabla \f d)$ instead of $F[\f d]$ to express that we consider $F_{OF}$ as a function in $\f d $ and $\nabla \f d$  instead of a functional in $\f d$.  
 This is a
plausible choice for an inhomogeneous (uniaxial) nematic liquid crystal. However, as proposed by de Gennes~\cite{gennes}, the theory
can also be used to describe the impact of a magnetic field, disregarding distortion effects. To this end,
he proposed the {\em ansatz}
\begin{equation}
F_H ( \f d ) = - \chi_\bot |\f H|^2  - (\chi_\|-\chi_\bot) ( \f d \cdot \f H )^2 \, ,
\end{equation}
which conforms with the ansatz (\ref{VH}) for the effective potential in a corresponding
stationary theory. The parameters $\chi_\|$ and $\chi_\bot$ are the magnetic susceptibility
constants for a magnetic field parallel and perpendicular to the director
$\f d$, respectively.

\subsection{Mathematical studies of the Erickson--Leslie theory}

From the physical side, most of the recent studies 
of nonequilibrium liquid crystals employ full $Q$-tensor theories rather than director dynamics since it is well established that the systems of interest are typically {\em not} uniaxial. From the mathematical perspective, however,
the lower-dimensional structure of the Erickson--Leslie theory,
which still carries important aspects of the orientational dynamics, makes this theory somewhat more accessible for rigorous treatments.
We thus summarize recent mathematical advances in this area.

To start with, we stress that the full system~\eqref{ELeq}-\eqref{oseen} of nonlinear partial differential equations involving nonlinear coupling between $\f v $ and $\f d$ as well
as algebraic restrictions is, from the mathematical point of view, rather difficult.
So far results  were only achieved under rather strong
simplifications. In particular, the second material derivative is nearly always neglected (see Ref.~\cite{inertia} for a first treatment).
 This seems to be appropriate since the constant $\rho_1 $, which appears as a prefactor of the second time derivative and is related to the moment of inertia,
  is typically negligible in a macroscopic viscous system.
Moreover, if this second-order term would not be neglected, the first equation would remain
of hyperbolic type, what would require a different mathematical approach to existence.


Several important contributions regarding the first-order Ericksen--Leslie theory are due to Lin and Liu~\cite{linliu1,linliu2,linliu3}.
In Ref.~\cite{linliu1} they considered a system,
where the (Oseen--Frank) free energy functional is written in the one-constant approximation (see Eq.~\eqref{oseen}). 
Further, the algebraic restriction $|\f d |=1$ is
incorporated by adding a Ginzburg--Landau penalty functional, a trick which was later copied
by several other authors.
The resulting free energy functional reads
\begin{align}
F_\varepsilon ( \f d , \nabla \f d) = \frac{1}{2} |\nabla \f d |^2 + \frac{1}{4\varepsilon^2}
 ( | \f d |^2 -1) ^2 .
\end{align}
In the dissipative stress tensor
$\sigma$ (see Eq.~(\ref{stressE})), Lin and Liu set all Leslie constants to zero except $\mu_4$, such that $\sigma=\mu_4 (\nabla \f v)_{\mathrm{sym}}$.
Further, the external forces appearing on the right-hand sides of Eqs.~\eqref{EL1}
and~\eqref{EL2} are neglected, $\lambda_1$ is set to one, and $\lambda_2$ is set to zero.
With this last simplification, translational forces of the fluid onto the director
are neglected, which enables  to prove a maximum principle for $|\f d|^2$
that is  essential for the later analysis. The simplified system considered in Lin and
Liu~\cite{linliu1} then reads
\begin{align}
\begin{split}
 \pat{\f d} {t}+ ( \f v \cdot \nabla ) \f d  & =  \Delta \f d
 -\frac{1}{\varepsilon^2}( |\f d |^2 -1 ) \f d \,,\\
 \pat{\f v}{t} + ( \f v \cdot \nabla ) \f v + \nabla p - \frac{\mu_4}{2} \Delta \f v
 & = -  \di \left ( \nabla \f d ^T \nabla \f d \right ),\\
 \di \f v &= 0 \,.
 \end{split}\label{simpleEL}
\end{align}
Remark that $\mu_4 \di ( \nabla \f v)_{\mathrm{sym}} = (\mu_4/2) \Delta \f v$ since $\f v$ is solenoidal, (\emph{i.e.}, $\di \f v=0$). 
For this system~\eqref{simpleEL},
Lin and Liu proved global existence of  weak solutions and local existence of
strong solutions (see Ref.~\cite{linliu1}). In Ref.~\cite{linliu3}, they succeeded in generalizing their results 
to a system in formal analogy to Eq.\eqref{simpleEL}, but with the full dissipative Leslie stress tensor
 $\sigma$ (see Eq.~\eqref{stressE}) and not neglecting the nonlinear coupling term $(\nabla \f v)_{\mathrm{skw}}\f d$ in Eq.~\eqref{EL1}.
This yields
\begin{align}
\begin{split}
 \pat{\f d} {t}+ ( \f v \cdot \nabla ) \f d  + (\nabla \f v)_{\mathrm{skw}} \f d + \frac{\lambda_2}{\lambda_1}
 (\nabla \f v)_{\mathrm{sym}} \f d & =  \Delta \f d
 -\frac{1}{\varepsilon^2}( |\f d |^2 -1 ) \f d \,,\\
 \pat{\f v}{t} + ( \f v \cdot \nabla ) \f v + \nabla p - \sigma'
 & = -  \di \left ( \nabla \f d ^T \nabla \f d \right ),\\
 \di \f v &= 0 \,.
 \end{split}\label{notsosimpleEL}
\end{align}
In treating this system, Lin and Liu~\cite{linliu3} again specialized on the case $\lambda_2=0$, because
the existence results  rely essentially on a maximum principle for $| \f d |^2$, which does not hold if $\lambda_2\neq 0$. To illustrate the application
of the maximum principle, we multiply the first 
member of Eqs.~\eqref{notsosimpleEL} (with $\lambda_2=0$) with  $2 \f d$, yielding formally an equation for $|\f d|^2$,
\begin{align}
\begin{split}
\frac{\partial}{\partial t} |\f d|^2 + ( \f v \cdot \nabla )|\f d |^2  =& 2 \Delta
\f d\cdot \f d -  \frac{2}{\varepsilon^2}\left ( |\f d|^2 -1 \right )|\f d |^2  \\=&
 \Delta |\f d|^2 - 2  | \nabla \f d |^2  -  \frac{2}{\varepsilon^2}\left (
|\f d|^2 -1 \right )|\f d |^2 .
\end{split}
\end{align}
Note that $\f d \cdot ( \nabla \f v)_{\text{skw}} \f d =0$.
A maximum of the function $|\f d|^2$ with respect to time and space
implies a maximum of  $|\f d|$ itself. Moreover, for a maximum of $| \f d|^2$, we observe that the first derivatives are zero, such  that $\partial_t | \f d|^2 +( \f v \cdot \nabla ) | \f d|^2=0$ and 
 $\Delta|\f d|^2\leq 0$. The assumption that the
maximum is attained at a point where the norm of $\f d$ is greater than $1$ leads
(because of the contribution of the Ginzburg--Landau penalization with $( |\f d|^2 -1 ) | \f
d|^2 >0$ for $|\f d|>1$) to a contradiction.
We further mention that Lin and Liu proved in~\cite{linliu2} partial regularity of weak solutions to system~\eqref{simple} implying  the convergence of the
  weak solutions of system~\eqref{simpleEL} to a measure-valued solution of the original
  problem with $|\f d |=1$ as $\varepsilon\ra 0$.

After the work of Lin and Liu, many other studies have appeared  
focusing on slightly more
complicated models. For example, Cavaterra, Rocca and Wu~\cite{allgemein}
considered model~\eqref{notsosimpleEL} with
$\lambda_2\neq 0$. They demonstrated
global existence of weak solutions without a maximum principle and, additionally, local existence of classical
solutions for periodic boundary conditions. 
Wu, Xu and Liu~\cite{hiu} studied the importance of  Parodi's
 relation~(last equation in \eqref{parodi}) for the well-posedness and stability of the system.
 The long-time
behaviour for a similar model was investigated
in Petzeltov{\'a}, Rocca and Schimperna~\cite{longtime}, and a nonisothermal model was considered in Feireisl,
 Rocca and Schimperna~\cite{isothermal}.

The restriction of the Erickson--Leslie theory regarding the {\em constant} magnitude of the director
has been revisited in an article by Wang, Zhang and Zhang~\cite{recent}. They derived, under specific assumptions for the constants appearing 
in the Leslie stress,
a reformulation of the Ericksen--Leslie system without the Ginzburg--Landau penalisation. For this system,
they derived various local well-posedness results, as well as global well-posedness for small
initial data.
We also mention a very recent generalization of the existence theory 
for the Erickson--Leslie system~\eqref{ELeq} by Emmrich and Lasarzik~\cite{unsere}. They showed
for a general class of free energies global existence of weak solutions.

%
%
\section{Cross-Linkages between the different models}\label{sect-rel}

Having discussed three major approaches towards the stationary and nonstationary behavior of
liquid crystals, one natural question concerns the linkages (if present) between the different models.
We start by recalling the (obvious) connections between the dynamical variables of interest.
Within the Doi--Hess theory, this is the probability density $f=f(\f x,\f n,t)$ (see, \emph{e.g.}, Eq.~(\ref{normed})), whose second moment
corresponds (up to a constant shift) to the  tensorial order parameter $Q=Q(\f x,t)$, see Eq.~\eqref{Q}. The eigenvector related
to the eigenvalue of the $Q$-tensor with the largest absolute value then represents the director $\f d=\f d(\f x,t)$ in the
Erickson--Leslie theory. 

In the case $\f v=0$, all of the approaches discussed here reduce to
minimization of a free energy functional. For the Doi--Hess model, one possible choice for the functional $F=F[f]$ is given in Eq.~\eqref{free}.
In the $Q$-tensor and Erickson--Leslie theory, the minimization is rather carried out with respect to the 
respective order parameter ($Q$ or $\f d$), and the free energies of interest are given
either by the Landau--de Gennes expression $F_{LG}$ (see~\eqref{landau} or alternatively the Maier--Saupe  expression~\eqref{Maier}), 
or by the  Oseen--Frank energy $F_{OF}$ (see~\eqref{oseen}) in the Ericksen--Leslie model Eq.~\eqref{ELeq}.

Are the solutions of these minimization procedures consistent? In this regard it is interesting to note that, according to
Majumdar and Zarnescu~\cite{majumdar2}, 
minimization of the Landau--de Gennes free energy~\eqref{landau}
within the one-constant approximation (and for vanishing elastic constant, \emph{i.e.}, $L_1\ra 0$) can be written in the form
\begin{align}
Q^* = s \left ( \f d ^* \otimes \f d^* - \frac{1}{3}I\right )\,.
\end{align}
Here, $\f d^*$ is a unit vector minimizing the Oseen--Frank free energy~\eqref{oseen} within the one-constant approximation, \emph{i.e.}, $k_1=k_2=k_3=\alpha$.

Not surprisingly, the situation for nonstationary systems is more involved.
For the case of a 
constant
velocity field 
and a spatially homogeneous director $\f d$,
Kuzuu and Doi~\cite{kuzuu} showed that the Doi--Onsager equation and the Ericksen--Leslie equation yield
the same results. This is since the time derivative and the convection term in Eqs.~\eqref{evov} and~\eqref{EL2} vanish, and 
the Ericksen term $\di \left ( \nabla \f d ^T \pat{F}{\nabla \f d }  \right )$ in Eq.~\eqref{EL2} becomes zero.

Beris and Edwards~\cite{beris} have studied the linkages between the
full $Q$-tensor theory and the Ericksen--Leslie.
As shown in their monograph~\cite[Sec.11.6.1]{beris}, the equation of motion~\eqref{evo} for the
$Q$-tensor reduces to equation~\eqref{ELeq} for the director in the uniaxial case. 

Based on this earlier work, Wang, Zhang and Zhang~\cite{zhang1} have generalized the derivation of the Ericksen--Leslie equations from the
Doi--Hess theory towards an inhomogeneous systems. To do this they formulated the problem as a limiting problem for vanishing Deborah number
using a Hilbert expansion.
The connection between the Doi--Hess model~\eqref{instatdoi}-\eqref{navdoi} and the nonstationary equations for the $Q$-tensor~\eqref{evo} was investigated by Han \emph{et\,al}.~\cite{han}. 
They derived the system of equations~\eqref{evo} from Eqs.~\eqref{instatdoi}-\eqref{navdoi} using the so-called Bingham closure (see Ref.~\cite{bingham}). 
The latter provides a strategy to approximate higher order moments of the probability density~$f$ based on the second moment of $f$ (\emph{i.e.}, $Q$).

  A rigorous derivation of the Ericksen Leslie system~\eqref{ELeq} with the Oseen--Frank 
  energy~\eqref{oseen} starting from the $Q$-tensor theory for vanishing elastic 
  constants~($L_i\ra 0$ for $i=1,2,3$ and $L_4\equiv 0$ in the Landau--de Gennes energy~\eqref{landau}) 
  was shown by Wang, Zhang and Zhang~\cite{rigorous} again by using a Hilbert expansion method.
   %
\section{Relative energy estimates for nonconvex energies\label{sect-rela}}

In view of the highly complex (and often nonlinear) evolution equations discussed in this paper, one may ask how different solutions (if present) can be related to one another.
For probability distribution functions, a modern concept to measure the distance between two solutions
involves the "relative entropy",
which is nowadays well established not only in the mathematical community but also in related fields such as information theory~\cite{kullback} and quantum entanglement~\cite{quantum}.
The concept goes back to Dafermos~\cite{dafermos}. For a strictly convex entropy function $\eta: \R \ra \R$, the relative entropy of two solutions~$u$ and $\tilde{u}$ is given by (see Ref.~\cite[Sec.~5.3]{dafermos2})
\begin{align}
\mathcal{E}[u|\tilde{u}]:= \eta [u] - \eta[\tilde{u}] - \eta'[\tilde{u}]( u -\tilde{u}) \,.\label{relencon}
\end{align}
The strict convexity of $\eta$ guarantees that $\mathcal{E}$ is nonnegative. 
Inserting, for example, the entropy function $\eta [u] =  u  \ln u - u $ (compare Eq.~\eqref{free}) gives for probability densities $ u $ and $\tilde{u}$ on a domain $\Omega$
\begin{align}
\begin{split}
\mathcal{E}[u|\tilde{u}] = {}&\int_{\Omega} \left ( u \ln u - u  - \tilde{u}\ln \tilde{u}+ \tilde u  - ( \ln \tilde{u}  ) ( u-\tilde{u})   \right )\de \f x = \int_{\Omega } u\ln \left (\frac{u}{\tilde{u}}\right ) \de \f x - \int_{\Omega} u\de \f x + \int_{\Omega} \tilde{u} \de \f x 
\\
={}& \int_{\Omega } u\ln \left (\frac{u}{\tilde{u}}\right ) \de \f x\,
\end{split}
\end{align}
since $\int_{\Omega} u\de \f x = \int_{\Omega} \tilde{u} \de \f x=0$.
This is the so-called Kullblack--Leibler divergence (see Ref.~\cite{leiber}). 
Importantly, the concept can also be used in a more general way to measure the distance between two possible solutions 
(not necessarily distribution functions) of an evolution equation based on a special distance function, the relative entropy.
This approach has been used, e.g., to show weak-strong uniqueness property of solutions (see {Feireisl and Novotn\'y}~\cite{novotny}), the stability of an equilibrium state (see Feireisl~\cite{feireislstab}), the convergence to a singular limit problem (see {Breit, Feireisl and Hofmanova}~\cite{breit} as well as {Feireisl}~\cite{feireislsingular}), or to derive {\em a posteriori} 
estimates of numerical solutions (see {Fischer}~\cite{fischer}).
Another possible application is the definition of a generalized solution concept, the dissipative solutions. The formulation of such a concept relies on an inequality instead of an equality (see Lions~\cite[Sec.~4.4]{lionsfluid}).
Weak-strong uniqueness means that, in case of the same initial and boundary values, a generalized solution coincides with the strong solution as long as the latter exists.

In this section, we present a new relative entropy approach 
for a simplified Erickson--Leslie model combined with the Oseen--Frank energy. Since the Oseen--Frank free energy is nonconvex, our approach generalizes the relative entropy approach towards nonconvex functions. In our context, the main quantities of interest are (free) energies rather than entropies. We therefore refer henceforth to a "relative energy" approach.
We exemplify the calculations for the Ericksen--Leslie model but it is likely that the approach be applied to the $Q$-tensor model as well. 

Consider the Ericksen--Leslie model~\eqref{ELeq} with $\rho_1\equiv 0$, $\lambda_1=1$, $\f F= 0$ and with the free energy 
\begin{align}
F[ \f d  ] = \frac{1}{2}  \int_\Omega |\nabla \f d| ^2 \de \f x +  \frac{1}{2} \int_{\Omega} |\sk d \f d | ^2\de \f x  \,.\label{simpenergy}
\end{align}
We remark that $ | \nabla \f d |^2 = ( \di \f d)^2 +( \f d \cdot \curl \f d)^2  +  |\f d \times  \curl \f d |^2 + \di ( \nabla \f d ^T \f d -( \di \f d) \f d   ) $ as long as $| \f d|=1$ and  $ 4 | \sk d \f d |^2 = | \f d \times \curl \f d |^2$. Thus, the above energy is  a simplification of the Oseen--Frank energy (see Eq.~\eqref{oseen}) with $ k_1 = k_2 = \alpha = 1/2$, $q=0$ and $ k_3 = 1/2+1/8$. We could as well handle the full system and the full Oseen--Frank energy. Here, we simplify the system to keep the calculations readable and to focus on the novelty regarding the nonconvex part of the free energy. 
In addition, we assume that a  constant velocity is prescribed, $\f v \equiv \text{const}$. 
The equation of motion~\eqref{ELeq} then simplifies to
\begin{align}
\begin{split}
\partial_t \f d + ( \f v \cdot \nabla ) \f d ={}&- \frac{\delta F}{\delta \f d }[\f d] 
= \di \left ( \frac{\partial F_{OF} }{\partial \nabla \f d }(\f d , \nabla \f d)  \right ) - \frac{\partial F_{OF} }{\partial \f d }( \f d ,\nabla \f d)\\
={}& \di \left ( \nabla \f d  + \left ( ( \nabla \f d)_{\skw} \f d \otimes \f d\right ) _{\skw}\right ) - ( \nabla \f d)_{\skw}^T (\nabla \f d)_{\skw} \f d   \\
= {}& \Delta \f d +  \di ( \sk d \f d \otimes \f d )_{\skw} -  \sk d^T \sk d \f d  \,.
\end{split}
\label{simple}
\end{align}
By $F_{OF}$ we denote the potential~\eqref{oseen} of the Oseen--Frank energy for the simplified case of Energy~\eqref{simpenergy}.
Since $\f v $ is constant, an integration by parts shows formally that
\begin{align}
\begin{split}
\int_{\Omega} (\f v \cdot \nabla ) \f d \cdot \frac{\delta F}{\delta \f d }[\f d]\de \f x  ={}& \int_{\Omega} (\f v \cdot \nabla ) \f d \cdot \left (-\di \left ( \frac{\partial F_{OF} }{\partial \nabla \f d }(\f d , \nabla \f d)  \right ) + \frac{\partial F_{OF} }{\partial \f d }( \f d ,\nabla \f d)\right )\de \f x\\ ={}&
\int_{\Omega} \left ((\f v \cdot \nabla ) \nabla \f d : \frac{\partial F_{OF} }{\partial \nabla \f d }(\f d , \nabla \f d)   + (\f v \cdot \nabla ) \f d \cdot  \frac{\partial F_{OF} }{\partial \f d }( \f d ,\nabla \f d)\right ) \de \f x \\
={}&\int_\Omega ( \f v \cdot \nabla ) F_{OF}( \f  d , \nabla \f d) \de \f x 
 =0
\end{split} 
 \label{divfrei}
\end{align}
for every solution $\f d$ of Eq.~\eqref{simple}.  
The last equation in the Calculation~\eqref{divfrei} holds since $\f v$ is assumed to be solenoidal and fulfill homogeneous Dirichlet boundary conditions.
Testing Eq.~\eqref{simple} with the variational derivative of the solution~$\f d$, \emph{i.e.}, multiplying Eq.~\eqref{simple} with $\delta F / \delta \f d$ and integrating over $\Omega \times (0,t)$, yields with $ \int_\Omega \partial_t \f d \cdot ( \delta F / \delta \f d) \de \f x = \partial_t F[\f d]$  formally the energy equality
\begin{align}
F[ \f d(t)] + \int_0^t  \int_\Omega \left | \frac{\delta F }{\delta \f d }[\f d(s)]\right |^2\de \f x  \de s= F[ \f d(0)]\quad \text{for }t\in (0,T) \,.\label{energyeq}
\end{align}  
The generalization of the concept of relative energies to nonconvex functionals relies on a suitable definition of the relative energy. If a functional $\eta$ is nonconvex, the relative energy defined by~\eqref{relencon} is not necessarily positive anymore. Here we introduce a new way to define the relative energy for nonconvex functions, which is nonnegative and allows to show an associated relative energy inequality (see inequality~\eqref{relEn}). The proof of this inequality is carried out in Proposition~\ref{prop}. 

The relative energy is given by 
\begin{align}
\mathcal{E}[\f d(t)| \dd(t)] := \frac{1}{2}\int_\Omega | \nabla \f d (t) - \nabla \dd (t) |^2\de \f x  + \frac{1}{2} \int_\Omega | ( \nabla \f d (t))_{\skw} \f d(t) - (\nabla \dd(t))_{\skw} \dd(t) |^2\de \f x \,.\label{relenergy}
\end{align}
 The following proposition gives an estimate of the relative energy and thus a measure of the distance of two solutions.
 \begin{proposition}\label{prop}
 Let $\f d$ and $ \dd $ be two sufficiently smooth solutions to~Eq.~\eqref{simple} fulfilling the same Dirichlet boundary conditions with boundary values that are constant in time. Then there holds for $t\in (0,T)$
 \begin{align}
 \begin{split}
 \mathcal{E}[\f d(t)| \dd(t)] \leq{}& 2 \Big ( \mathcal{E}[\f d(0)| \dd(0)] + \max_{( \f x ,t) \in \overline{\Omega}\times [0,T]}|( \nabla \dd (\f x ,t)  )_{\mathrm{skw}} \dd(\f x,t)|^2 \int_{\Omega}| \f d(0)-\dd(0)|^2\de \f x  
\\ &+  \int_\Omega \left |  ( \f d(0)   - \dd (0) ) \cdot  ( ( \nabla \f d (0)  )_{\skw}-( \nabla \dd(0)   )_{\skw})^T( \nabla \dd(0)   )_{\skw} \dd(0) \right |   \de \f x 
 \Big ) 
 \exp\left ( \mathcal{K}(t)\right )\,,
 \end{split}\label{relEn}
\end{align}  
 where
 \begin{align}
 \begin{split}
  \mathcal{K}(t)={}& 
4 c t  \max_{( \f x ,t) \in \overline{\Omega}\times [0,T]}|( \nabla \dd (\f x ,t)  )_{\skw} \dd(\f x,t)|^4  \\ &+ 
2(1+c)\int_0^t\max_{\f x \in \overline{\Omega}} \left (
| \dd(\f x )| | \nabla \partial_t \dd(\f x )| + | \nabla \dd(\f x )| | \partial_t \dd(\f x )|+ | \nabla \partial_t \dd(\f x)|    +  | \partial_t \dd(\f x)|+ | \f v (\f x) |^2  \right )\de s 
%
\,.
\end{split}
\label{K}
 \end{align}
Both solutions coincide if $\f d (0)= \dd(0)$. Here $c$ denotes the constant of the Poincar\'e inequality (see Eq.~\eqref{poincare}).
 \end{proposition}
\begin{proof}
The relative energy~\eqref{relenergy} can be explicitly calculated 
using the binomial formula 
\begin{align}
\mathcal{E}[\f d(t)| \dd(t)] = F[\f d (t)] + F[ \dd(t)] - \int_\Omega  \nabla \f d (t): \nabla \dd(t) \de \f x  - \int_\Omega \f d (t)\cdot ( \nabla \f d(t))_{\skw}^T  ( \nabla \dd(t))_{\skw} \dd (t) \de \f x \,.\label{calcrel}
\end{align}
Similarly, we calculate the difference of the variational derivatives using the binomial formula
\begin{multline}
\int_0^t \int_\Omega \left | \frac{\delta F }{\delta \f d }[\f d] - \frac{\delta F }{\delta \f d }[\dd] \right |^2 \de \f x \de s   =\\ \int_0^t \int_\Omega \left | \frac{\delta F }{\delta \f d }[\f d]  \right |^2 \de \f x\de s + \int_0^t \int_\Omega \left |  \frac{\delta F }{\delta \f d }[\dd] \right |^2 \de \f x \de s - 2 \int_0^t \int_\Omega \frac{\delta F }{\delta \f d }[\f d] \cdot  \frac{\delta F }{\delta \f d }[\dd]  \de \f x \de s \label{reldiss}
\end{multline} 
Adding Eq.~\eqref{calcrel} and Eq.~\eqref{reldiss} and using the energy equality~\eqref{energyeq} for the two solutions gives
\begin{align}
\begin{split}
\mathcal{E}[\f d(t)| \dd(t)]+\int_0^t \int_\Omega \left | \frac{\delta F }{\delta \f d }[\f d] - \frac{\delta F }{\delta \f d }[\dd] \right |^2 \de \f x \de s ={}& F[\f d(0)] + F [\dd(0)] - \int_\Omega  \nabla \f d (t): \nabla \dd(t) \de \f x \\& - \int_\Omega \f d (t)\cdot ( \nabla \f d(t))_{\skw}^T ( \nabla \dd(t))_{\skw} \dd (t) \de \f x
\\&
- 2 \int_0^t \int_\Omega \frac{\delta F }{\delta \f d }[\f d] \cdot  \frac{\delta F }{\delta \f d }[\dd]  \de \f x \de s 
\end{split}\label{sum}
\end{align}

In the following we prove two integration by parts formulas (Eq.~\eqref{intpart1} and Eq.~\eqref{intpart5}), which are essential to estimate the relative energy~\eqref{relenergy}.
In regard of the third  term on the right-hand side of Eq.~\eqref{sum}, we observe with the fundamental theorem of calculus and an integration by parts that
\begin{align}
\begin{split}
 \int_{\Omega} \nabla \f d (t):  \nabla \dd(t)\de \f x   - \int_{\Omega} \nabla \f d (0):  \nabla \dd(0) \de \f x  = {}& \int_0^t \partial_s \int_\Omega  \nabla \f d(s) : \nabla \dd(s) \de \f x \de s \\
={}&  \int_0^t  \int_\Omega  \left ( \nabla \partial _s \f d(s) : \nabla \dd(s) + \nabla \f d (s) : \nabla \partial _s \dd(s) \right )  \de \f x \de s \\
={}&\int_0^t \int_{\Omega}( \partial_s \f d (s) , - \Delta \dd(s) ) +  ( \partial_s \dd(s) , - \Delta \f d (s))\de \f x  \de s \,.
\end{split}
\label{intpart1}
\end{align}
The boundary terms disappear since $ \partial_s \f d $ and $\partial_s \dd$ vanish at the boundary. This is due to the fact that the prescribed boundary values are constant in time.  
Similarly but somehow more involved, we obtain in regard of the fourth term on the right hand side of Eq.~\eqref{sum} that
\begin{align}
\begin{split}
&\int_\Omega  \f d (t)  \cdot ( \nabla \f d(t))_{\skw}^T ( \nabla \dd(t))_{\skw} \dd (t) \de \f x  -  \int_\Omega   \f d (0) \cdot ( \nabla \dd(0))_{\skw} ^T ( \nabla \dd(0))_{\skw} \dd (0) \de \f x \\
&{} = \int_0^t \partial _s \int_{\Omega}   \f d (s)  \cdot( \nabla \f d(s))_{\skw}^T ( \nabla \dd(s))_{\skw} \dd (s) \de \f x \de s 
\\
&{} = \int_0^t  \int_{\Omega}  \left ( \f d (s)  \cdot ( \nabla \partial _s\f d(s))_{\skw}^T( \nabla \dd(s))_{\skw} \dd (s)+  \partial _s\f d (s)  \cdot( \nabla \f d(s))_{\skw}^T ( \nabla \dd(s))_{\skw} \dd (s) \right ) \de \f x \de s 
\\
&{} \quad+ \int_0^t  \int_{\Omega}  \left ( \f d (s)  \cdot( \nabla \f d(s))_{\skw}^T ( \nabla \partial _s\dd(s))_{\skw} \dd (s)+  \f d (s)  \cdot ( \nabla \f d(s))_{\skw}^T ( \nabla \dd(s))_{\skw} \partial _s\dd (s) \right ) \de \f x \de s 
\,.
\end{split} \label{intpart2}
\end{align}
The right-hand side of Eq.~\eqref{intpart2} is rearranged further on. In the following, we omit the dependence on the integration parameter $s$ for brevity. 
Adding and simultaneously subtracting the terms 
\begin{align}
\begin{split}
&{}  \int_0^t  \int_{\Omega}  \left ( \dd    \cdot( \nabla \partial _s\f d )_{\skw}^T ( \nabla \dd )_{\skw} \dd  +  \partial _s\f d    \cdot ( \nabla \dd )_{\skw}^T ( \nabla \dd )_{\skw} \dd   \right ) \de \f x \de s 
\\
&{} \quad+ \int_0^t  \int_{\Omega}  \left ( \f d    \cdot( \nabla \f d )_{\skw}^T ( \nabla \partial _s\dd )_{\skw} \f d  +  \f d    \cdot ( \nabla \f d )_{\skw}^T ( \nabla \f d )_{\skw} \partial _s\dd   \right ) \de \f x \de s 
\end{split}
\end{align}
leads to
\begin{align}
\begin{split}
&\int_\Omega   \f d (t)  \cdot ( \nabla \f d(t))_{\skw}^T ( \nabla \dd(t))_{\skw} \dd (t) \de \f x  -  \int_\Omega  \f d (0) \cdot( \nabla \f d(0))_{\skw} ^T  ( \nabla \dd(0))_{\skw} \dd (0) \de \f x \\
&{} = \int_0^t  \int_{\Omega}  \left ( \dd    \cdot( \nabla \partial _s\f d )_{\skw}^T ( \nabla \dd )_{\skw} \dd  +  \partial _s\f d    \cdot( \nabla \dd )_{\skw}^T ( \nabla \dd )_{\skw} \dd   \right ) \de \f x \de s 
\\
&{} \quad+ \int_0^t  \int_{\Omega}  \left (\f d    \cdot( \nabla \f d )_{\skw} ^T ( \nabla \partial _s\dd )_{\skw} \f d  +  \f d    \cdot( \nabla \f d )_{\skw}^T ( \nabla \f d )_{\skw} \partial _s\dd   \right ) \de \f x \de s 
\\
&{} \quad+ \int_0^t  \int_{\Omega}  \left ( (\f d  - \dd )  \cdot ( \nabla \partial _s\f d )_{\skw}^T ( \nabla \dd )_{\skw} \dd  +  \partial _s\f d    \cdot(( \nabla \f d )_{\skw}- ( \nabla \dd )_{\skw})^T ( \nabla \dd )_{\skw} \dd   \right ) \de \f x \de s 
\\
&{} \quad+ \int_0^t  \int_{\Omega}  \left ( \f d    \cdot ( \nabla \f d )_{\skw}^T ( \nabla \partial _s\dd )_{\skw} (\dd  - \f d  ) + \f d    \cdot ( \nabla \f d )_{\skw} ^T ( ( \nabla \dd )_{\skw} - ( \nabla \f d )_{\skw}) \partial _s\dd   \right ) \de \f x \de s 
\,.
\end{split} \label{intpart3}
\end{align}
For the first two terms of the right-hand side of Eq.~\eqref{intpart3}, an integration by parts shows that
\begin{align}
\begin{split}
&\int_0^t  \int_{\Omega}  \left ( \dd    \cdot( \nabla \partial _s\f d )_{\skw}^T  ( \nabla \dd )_{\skw} \dd  +  \partial _s\f d    \cdot ( \nabla \dd )_{\skw}^T ( \nabla \dd )_{\skw} \dd   \right ) \de \f x \de s 
\\
&{} + \int_0^t  \int_{\Omega}  \left (\f d    \cdot ( \nabla \f d )_{\skw} ^T  ( \nabla \partial _s\dd )_{\skw} \f d  +  \f d    \cdot ( \nabla \f d )_{\skw}^T ( \nabla \f d )_{\skw} \partial _s\dd   \right ) \de \f x \de s 
\\
{}&{}= \int_0^t \int_{\Omega}  \partial_s \f d  \cdot  \left (-  \di ( ( \nabla \dd)_{\skw} \dd \otimes \dd )_{\skw} +  ( \nabla \dd)_{\skw} ^T ( \nabla \dd)_{\skw} \dd\right )   \de \f x  \de s\\&\quad+ \int_0^t \int_\Omega  \partial_s \dd  \cdot \left (  -\di ( \sk d \f d \otimes \f d )_{\skw} +  \sk d^T \sk d \f d  \right )\de \f x  \de s\,.
\end{split}\label{intpart4}
\end{align}
Note that the boundary terms disappear since $ \partial_s \f d $ and $\partial_s \dd$ vanish at the boundary (compare to Eq.~\eqref{intpart1}).
The last two terms  on the right-hand side of Eq.~\eqref{intpart3}
are rearranged by adding and subtracting the term
\begin{align}
 \int_0^t  \int_{\Omega}  \left ((\f d  - \dd )  \cdot ( \nabla \partial _s\dd )_{\skw} ^T( \nabla \dd )_{\skw} \dd  +  \partial _s\dd    \cdot (( \nabla \f d )_{\skw}- ( \nabla \dd )_{\skw})^T ( \nabla \dd )_{\skw} \dd   \right ) \de \f x \de s \,.
\end{align}
This leads to 
\begin{align}
\begin{split}
&{} \int_0^t  \int_{\Omega}  \left ( (\f d  - \dd )  \cdot ( \nabla \partial _s\f d )_{\skw}^T ( \nabla \dd )_{\skw} \dd  +  \partial _s\f d    \cdot (( \nabla \f d )_{\skw}- ( \nabla \dd )_{\skw})^T  ( \nabla \dd )_{\skw} \dd   \right ) \de \f x \de s 
\\
&{} + \int_0^t  \int_{\Omega}  \left ( \f d    \cdot ( \nabla \f d )_{\skw}^T ( \nabla \partial _s\dd )_{\skw} (\dd  - \f d  ) +  \f d    \cdot ( \nabla \f d )_{\skw}^T ( ( \nabla \dd )_{\skw} - ( \nabla \f d )_{\skw}) \partial _s\dd   \right ) \de \f x \de s 
\\
&{}= 
 \int_0^t  \int_{\Omega}  \left (   \partial _s(( \nabla \f d )_{\skw}- ( \nabla \dd )_{\skw}) (\f d  - \dd )   + (( \nabla \f d )_{\skw}- ( \nabla \dd )_{\skw}) \partial _s(\f d- \dd)      \right )  \cdot ( \nabla \dd )_{\skw} \dd \de \f x \de s 
\\
&{} \quad+ \int_0^t  \int_{\Omega} ( ( \nabla \f d )_{\skw} \f d -  ( \nabla \dd )_{\skw} \dd)   \cdot  \left (  ( \nabla \partial _s\dd )_{\skw} (\dd  - \f d  ) +  ( ( \nabla \dd )_{\skw} - ( \nabla \f d )_{\skw}) \partial _s\dd   \right ) \de \f x \de s \,.
\end{split}\label{absch1}
\end{align}
Using the product rule, the first term on the right-hand side of Eq.~\eqref{absch1} can be expressed via
\begin{align}
\begin{split}
 &\int_0^t  \int_{\Omega}  \left (   \partial _s(( \nabla \f d )_{\skw}- ( \nabla \dd )_{\skw}) (\f d  - \dd )   + (( \nabla \f d )_{\skw}- ( \nabla \dd )_{\skw}) \partial _s(\f d- \dd)      \right )  \cdot ( \nabla \dd )_{\skw} \dd \de \f x \de s \\
&{}=  \int_0^t  \int_{\Omega}     \partial _s\left (( \nabla \f d )_{\skw}- ( \nabla \dd )_{\skw}) (\f d  - \dd )\right )    \cdot ( \nabla \dd )_{\skw} \dd \de \f x \de s 
 \\
 &{}=  \int_0^t \partial_s \int_\Omega \left ( ( ( \nabla \f d   )_{\skw}-( \nabla \dd   )_{\skw}) ( \f d   - \dd  ) \right ) \cdot  ( \nabla \dd   )_{\skw} \dd    \de \f x  \de s \\
& \quad- \int_0^t \int_{\Omega} \left (( ( \nabla \f d   )_{\skw}-( \nabla \dd   )_{\skw}) ( \f d   - \dd  )\right ) \cdot   \partial_s \left (( \nabla \dd   )_{\skw} \dd  \right )   \de \f x  \de s\,.
\end{split}\label{absch2}
\end{align}
The fundamental theorem of calculus gives for the first term on the right-hand side of Eq.~\eqref{absch2} that
\begin{align}
\begin{split}
&\int_0^t \partial_s \int_\Omega  \left (( ( \nabla \f d   )_{\skw}-( \nabla \dd   )_{\skw}) ( \f d   - \dd  )\right ) \cdot  ( \nabla \dd   )_{\skw} \dd    \de \f x  \de s \\
&{}=   \int_{\Omega}  ( \f d(t)   - \dd (t) ) \cdot  ( ( \nabla \f d (t)  )_{\skw}-( \nabla \dd(t)   )_{\skw})^T  ( \nabla \dd(t)   )_{\skw} \dd(t)   \de \f x  
\\
{}&\quad -  \int_{\Omega} ( \f d(0)   - \dd (0) ) \cdot ( ( \nabla \f d (0)  )_{\skw}-( \nabla \dd(0)   )_{\skw}) ^T  ( \nabla \dd(0)   )_{\skw} \dd(0)   \de \f x \,.
\end{split}\label{absch3}
\end{align}
Putting Eqs.~\eqref{intpart4},~\eqref{absch1},~\eqref{absch2} and~\eqref{absch3} back into~\eqref{intpart3} finally gives
\begin{align}
\begin{split}
\int_\Omega & \f d (t)  \cdot ( \nabla \f d(t))_{\skw} ^T ( \nabla \dd(t))_{\skw} \dd (t) \de \f x  -  \int_\Omega   \f d (0) \cdot ( \nabla \f d(0))_{\skw}^T  ( \nabla \dd(0))_{\skw} \dd (0) \de \f x \\
=& \int_0^t \int_{\Omega}  \partial_s \f d  \cdot  \left (-  \di ( ( \nabla \dd)_{\skw} \dd \otimes \dd )_{\skw} +  ( \nabla \dd)_{\skw} ^T ( \nabla \dd)_{\skw} \dd\right )   \de \f x  \de s
\\
&+ \int_0^t \int_\Omega  \partial_s \dd  \cdot \left (  -\di ( \sk d \f d \otimes \f d )_{\skw} +  \sk d^T \sk d \f d  \right )\de \f x  \de s
\\
& + \int_{\Omega}  ( \f d(t)   - \dd (t) ) \cdot ( ( \nabla \f d (t)  )_{\skw}-( \nabla \dd(t)   )_{\skw})^T  ( \nabla \dd(t)   )_{\skw} \dd(t)   \de \f x  
\\
&{}-  \int_{\Omega}  ( \f d(0)   - \dd (0) ) \cdot ( ( \nabla \f d (0)  )_{\skw}-( \nabla \dd(0)   )_{\skw})^T   ( \nabla \dd(0)   )_{\skw} \dd(0)   \de \f x \\
& - \int_0^t \int_{\Omega} \left ( ( ( \nabla \f d   )_{\skw}-( \nabla \dd   )_{\skw}) ( \f d   - \dd  )\right ) \cdot   \partial_s \left (( \nabla \dd   )_{\skw} \dd  \right )   \de \f x  \de s\\
& + \int_0^t \int_{\Omega} (    ( \nabla \f d   )_{\skw}\f d -( \nabla \dd   )_{\skw}\dd   ) \cdot \left (( \nabla \partial _s \dd )_{\skw} ( \dd   -\f d ) + ( ( \nabla \dd   )_{\skw}-( \nabla \f d   )_{\skw})  \partial _s \dd  \right )  \de \f x  \de s \,.
\end{split} \label{intpart5}
\end{align}
This is the integration by  parts formula needed to estimate the relative energy.

The right-hand side of Eq.~\eqref{sum} 
is now estimated using 
  Eq.~\eqref{intpart1} and Eq.~\eqref{intpart5} as well as Young's inequality (%
see Ref.~\cite[Section~4.8]{inequalities}), as follows:
\begin{align}
\begin{split}
\mathcal{E}&[\f d (t) | \dd(t) ] +   \int_0^t \int_\Omega \left | \frac{\delta F }{\delta \f d }[\f d] - \frac{\delta F }{\delta \f d }[\dd] \right |^2 \de \f x  \de s 
\\ 
\leq {}&\mathcal{E}[\f d(0)| \dd(0)] - \int_0^t \int_\Omega \left (\partial _s \f d +  \frac{\delta F }{\delta \f d }[\f d] \right  ) \cdot \frac{\delta F }{\delta \f d }[\dd]\de \f x  \de s 
 - \int_0^t \int_\Omega \left ( \partial _s \dd +  \frac{\delta F }{\delta \f d }[\dd] \right ) \cdot \frac{\delta F }{\delta \f d }[\f d]\de \f x \de s \\
&
 -   \int_\Omega  ( \f d(t)   - \dd (t) ) \cdot ( ( \nabla \f d (t)  )_{\skw}-( \nabla \dd(t)   )_{\skw})^T ( \nabla \dd(t)   )_{\skw} \dd(t)    \de \f x 
  \\
  &+   \int_\Omega  ( \f d(0)   - \dd (0) ) \cdot ( ( \nabla \f d (0)  )_{\skw}-( \nabla \dd(0)   )_{\skw})^T ( \nabla \dd(0)   )_{\skw} \dd(0)    \de \f x 
  \\
&+\frac{1}{2} \int_0^t \max_{\f x \in \overline{\Omega}} \left (
| \dd(\f x )| | \nabla \partial_s \dd(\f x )| + | \nabla \dd(\f x )| | \partial_s \dd(\f x )| \right ) \int_{\Omega} \left (| ( \nabla \f d   )_{\skw}-( \nabla \dd   )_{\skw}|^2 + | \f d   - \dd  |^2  \right )\de \f x\de s\\
&+ \frac{1}{2}\int_0^t     \max_{\f x \in \overline{\Omega}}| \nabla \partial_s \dd(\f x)|    \int_\Omega \left (| \f d   - \dd  |^2 +|( \nabla \f d   )_{\skw}\f d -( \nabla \dd   )_{\skw}\dd  |^2 \right ) \de \f x \de s 
\\
&+\frac{1}{2} \int_0^t     \max_{\f x \in \overline{\Omega}}  | \partial_s \dd(\f x)| \int_\Omega \left ( | ( \nabla \f d   )_{\skw}-( \nabla \dd   )_{\skw}|^2 +|( \nabla \f d   )_{\skw}\f d -( \nabla \dd   )_{\skw}\dd  |^2 \right ) \de \f x \de s 
\,.
\end{split}\label{Ecalc}
\end{align}
The terms in the three last lines of Eq.~\eqref{Ecalc} are estimates of the terms in the two last lines of Eq.~\eqref{intpart5}.
To get the right-hand side of inequality~\eqref{Ecalc}, we explicitly used that Eq.~\eqref{calcrel} also holds for $t=0$ and that the variational derivative is given by
\begin{align*}
\frac{\delta F}{\delta \f d }[\f d] = - \Delta \f d -  \di ( \sk d \f d \otimes \f d )_{\skw} +  \sk d^T \sk d \f d
\end{align*} 
and similarly for $\dd$. 
Inserting Eq.~\eqref{simple}, using that $\f v$ is the same for both solutions%
, yields for the second and third term on the right-hand side of inequality~\eqref{Ecalc} that
\begin{align}
\begin{split}
 -{}&\int_0^t \int_\Omega \left (\partial _t \f d +  \frac{\delta F }{\delta \f d }[\f d] \right  ) \cdot \frac{\delta F }{\delta \f d }[\dd]\de \f x  \de s 
 - \int_0^t \int_\Omega \left ( \partial _t \dd +  \frac{\delta F }{\delta \f d }[\dd] \right ) \cdot \frac{\delta F }{\delta \f d }[\f d]\de \f x \de s  \\
={}&   \int_0^t \int_\Omega( \f v \cdot \nabla ) \f d   \cdot  \frac{\delta F }{\delta \f d }[\dd]\de \f x  \de s 
+  \int_0^t \int_\Omega ( \f v \cdot \nabla) \dd  \cdot \frac{\delta F }{\delta \f d }[\f d]\de \f x  \de s \,.
  \end{split}\label{eqeins}
\end{align}
 Due to Eq.~\eqref{divfrei}, we observe for the right-hand side of Eq.~\eqref{eqeins} that
 \begin{align}
 \begin{split}
&  \int_0^t \int_\Omega( \f v \cdot \nabla ) \f d   \cdot \frac{\delta F }{\delta \f d }[\dd]\de \f x  \de s 
 + \int_0^t \int_{\Omega} ( \f v \cdot \nabla) \dd  \cdot \frac{\delta F }{\delta \f d }[\f d]\de \f x  \de s \\
  {}& {}= \int_0^t \int_{\Omega}( \f v \cdot \nabla ) \f d   \cdot \left ( \frac{\delta F }{\delta \f d }[\dd]- \frac{\delta F }{\delta \f d }[\f d]\right ) \de \f x \de s 
+ \int_0^t \int_{\Omega} ( \f v \cdot \nabla) \dd  \cdot \left (\frac{\delta F }{\delta \f d }[\f d]- \frac{\delta F }{\delta \f d }[\dd]\right ) \de\f x \de s \\
  {}&{}= \int_0^t \int_\Omega  \f v \cdot \left ( \nabla \f d - \nabla \dd \right ) ^T \left ( \frac{\delta F }{\delta \f d }[\f d]-   \frac{\delta F }{\delta \f d }[\dd] \right ) \de \f x \de s\,.
 \end{split}\label{eqzwei}
 \end{align}
 Hence, the right-hand side of Eq.~\eqref{eqeins} can be estimated using  Young's inequality  (see Ref.~\cite[Section~4.8]{inequalities}), 
 \begin{align}
\begin{split}
&  \int_0^t \int_{\Omega}  \f v \cdot  \left ( \nabla \f d - \nabla \dd \right )^T \left (\frac{\delta F }{\delta \f d }[\f d]-   \frac{\delta F }{\delta \f d }[\dd] \right )\de \f x  \de s \\
 &\leq{}   \frac{1}{2}\int_0^t  \max_{\f x \in \overline{\Omega}}| \f v (\f x) |^2  \int_\Omega | \nabla \dd - \nabla \f d| ^2 \de \f  x \de s + \frac{1}{2}\int_0^t  \int_{\Omega}\left | \frac{\delta F }{\delta \f d }[\f d]-   \frac{\delta F }{\delta \f d }[\dd] \right |^2 \de \f x  \de s\,.
 \end{split}\label{eqdrei}
\end{align}
For the fourth term on the right-hand side of inequality~\eqref{Ecalc}, we observe by Young's inequality  (see Ref.~\cite[Section~4.8]{inequalities})
\begin{multline}
- \int_\Omega ( \f d(t)   - \dd (t) ) \cdot  ( ( \nabla \f d (t)  )_{\skw}-( \nabla \dd(t)   )_{\skw}) ^T  ( \nabla \dd(t)   )_{\skw} \dd(t)    \de \f x  \\
 \leq {} \frac{1}{4 } \int_\Omega| \nabla \f d (t) - \nabla \dd(t)|^2\de\f x +  \max_{( \f x ,t) \in \overline{\Omega}\times [0,T]}|  ( \nabla \dd (\f x ,t)  )_{\skw} \dd(\f x,t)|^2 \int_\Omega| \f d(t)   - \dd (t)|^2\de \f x \,.\label{pktwab}
\end{multline}
With respect to the last term in Eq.~\eqref{pktwab}, we observe with the fundamental theorem of calculus and Eq.~\eqref{simple} that 
\begin{align}
\begin{split}
\int_{\Omega}| \f d(t)   - \dd (t)|^2\de \f x  -{} \int_{\Omega}| \f d (0) - \dd(0)|^2\de \f x 
={}& \int_0^t \partial_s \int_\Omega | \f d - \dd|^2\de \f x \de s
 \\
= 
{}&  
 2 \int_0^t  \int_\Omega (\partial_s \f d - \partial_s \dd )\cdot ( \f d - \dd) \de \f x \de s
\\
={}&
- 2 \int_0^t  \int_\Omega ( \f v \cdot \nabla) (\f d-  \dd) \cdot ( \f d - \dd) \de \f x \de s
\\ &
 -2 \int_0^t \int_\Omega \left (\frac{\delta F }{\delta \f d }[\f d]-   \frac{\delta F }{\delta \f d }[\dd] \right ) \cdot  \left (  \f d - \dd     \right )\de \f x \de s \,.
 \end{split}\label{absch5}
\end{align}
The first term on the right-hand side of Eq.~\eqref{absch5} vanishes since $\f v$ is solenoidal and fulfills homogeneous Dirichlet boundary conditions.
With Young's inequality (see Ref.~\cite[Section~4.8]{inequalities}), we find
\begin{multline}
2\max_{( \f x ,t) \in \overline{\Omega}\times [0,T]}|( \nabla \dd (\f x ,t)  )_{\skw} \dd(\f x,t)|^2  \int_0^t\int_{\Omega} \left (\frac{\delta F }{\delta \f d }[\f d]-\frac{\delta F }{\delta \f d }[\dd] \right ) \cdot \left ( \f d - \dd     \right )\de \f x \de s 
\\ \leq  \frac{1}{2}\int_0^t \int_{\Omega} \left | \frac{\delta F }{\delta \f d }[\f d]-   \frac{\delta F }{\delta \f d }[\dd] \right |^2 \de s +2  \max_{( \f x ,t) \in \overline{\Omega}\times [0,T]}|( \nabla \dd (\f x ,t)  )_{\skw} \dd(\f x,t)|^4 \int_0^t  \int_{\Omega}| \f d   - \dd |^2\de \f x  \de s  \,. \label{absch6}
\end{multline}
Inserting Eqs.~\eqref{eqeins}, \eqref{eqzwei}, \eqref{eqdrei}, \eqref{pktwab} and~\eqref{absch6} 
back into Eq.~\eqref{Ecalc} gives
\begin{align}
\begin{split}
\mathcal{E}&[\f d (t) | \dd(t) ] +   \int_0^t \int_\Omega \left | \frac{\delta F }{\delta \f d }[\f d] - \frac{\delta F }{\delta \f d }[\dd] \right |^2 \de \f x  \de s
\\\leq{} &  \mathcal{E}[\f d(0)| \dd(0)]  
 + \frac{1}{4 } \int_\Omega| \nabla \f d (t) - \nabla \dd(t)|^2\de\f x + \frac{1}{2}\int_0^t \int_{\Omega} \left | \frac{\delta F }{\delta \f d }[\f d]-   \frac{\delta F }{\delta \f d }[\dd] \right |^2 \de s\\&  + 2\max_{( \f x ,t) \in \overline{\Omega}\times [0,T]}|( \nabla \dd (\f x ,t)  )_{\skw} \dd(\f x,t)|^4 \int_0^t  \int_{\Omega}| \f d   - \dd |^2\de \f x  \de s
 \\&+ \max_{( \f x ,t) \in \overline{\Omega}\times [0,T]}|( \nabla \dd (\f x ,t)  )_{\skw} \dd(\f x,t)|^2 \int_\Omega | \f d(0) - \dd(0)|^2 \de \f x 
  \\
  & +    \int_\Omega  ( \f d(0)   - \dd (0) ) \cdot ( ( \nabla \f d (0)  )_{\skw}-( \nabla \dd(0)   )_{\skw})^T ( \nabla \dd(0)   )_{\skw} \dd(0)    \de \f x 
  \\
& +\frac{1}{2} \int_0^t  \max_{\f x \in \overline{\Omega}}| \f v (\f x) |^2  \int_\Omega | \nabla \dd - \nabla \f d| ^2 \de \f  x \de s + \frac{1}{2}\int_0^t  \int_{\Omega}\left | \frac{\delta F }{\delta \f d }[\f d]-   \frac{\delta F }{\delta \f d }[\dd] \right |^2 \de \f x  \de s
\\
&+\frac{1}{2} \int_0^t \max_{\f x \in \overline{\Omega}} \left (
| \dd(\f x )| | \nabla \partial_s \dd(\f x )| + | \nabla \dd(\f x )| | \partial_s \dd(\f x )| \right ) \int_{\Omega} \left (| ( \nabla \f d   )_{\skw}-( \nabla \dd   )_{\skw}|^2 + | \f d   - \dd  |^2  \right )\de \f x\de s\\
&+\frac{1}{2} \int_0^t     \max_{\f x \in \overline{\Omega}}| \nabla \partial_s \dd(\f x)|    \int_\Omega \left (| \f d   - \dd  |^2 +|( \nabla \f d   )_{\skw}\f d -( \nabla \dd   )_{\skw}\dd  |^2 \right ) \de \f x \de s 
\\
&+ \frac{1}{2}\int_0^t     \max_{\f x \in \overline{\Omega}}  | \partial_s \dd(\f x)| \int_\Omega \left ( | ( \nabla \f d   )_{\skw}-( \nabla \dd   )_{\skw}|^2 +|( \nabla \f d   )_{\skw}\f d -( \nabla \dd   )_{\skw}\dd  |^2 \right ) \de \f x \de s 
\,.
\end{split}\label{Ecalc2}
\end{align}
The Poincar\'e inequality ensures that there exists a constant $c>0$, such that
\begin{align}
\int_{\Omega}| \f d   - \dd |^2 \de \f x \leq c \int_{\Omega}|\nabla  \f d   - \nabla \dd |^2 \de \f x\,.\label{poincare}
\end{align}
Note that the Poincar\'e inequality (see Ref.~\cite{poincare}) is applicable since $\f d$ and $\dd$ fulfill the same Dirichlet boundary conditions. 
Using the definition~\eqref{relenergy} of the relative energy shows that
\begin{align}
\begin{split}
\frac{1}{2}\mathcal{E}&[\f d (t) | \dd(t) ]  \de s \leq \mathcal{E}[\f d(0)| \dd(0)] 
 +   \int_\Omega \left | ( \f d(0)   - \dd (0) ) \cdot ( ( \nabla \f d (0)  )_{\skw}-( \nabla \dd(0)   )_{\skw})^T ( \nabla \dd(0)   )_{\skw} \dd(0)  \right |   \de \f x 
 \\&  + \max_{( \f x ,t) \in \overline{\Omega}\times [0,T]}|( \nabla \dd (\f x ,t)  )_{\skw} \dd(\f x,t)|^2 \int_\Omega | \f d(0) - \dd(0)|^2 \de \f x  \\
&  + 2 c \max_{( \f x ,t) \in \overline{\Omega}\times [0,T]}|( \nabla \dd (\f x ,t)  )_{\skw} \dd(\f x.t)|^4 \int_0^t \mathcal{E}[\f d  | \dd ]  \de s
\\
&+ (1+c)\int_0^t \max_{\f x \in \overline{\Omega}} \left (
| \dd(\f x )| | \nabla \partial_s \dd(\f x )| + | \nabla \dd(\f x )| | \partial_s \dd(\f x )|+ | \nabla \partial_s \dd(\f x)|    +  | \partial_s \dd(\f x)|+ | \f v (\f x) |^2  \right ) \mathcal{E}[\f d  | \dd ]\de s\,.
\end{split}\label{Ecalc3}
\end{align}
Gronwall's estimate (see Ref.~\cite{gronwall}) yields the asserted inequality~\eqref{relEn}.
\end{proof} 
 Proposition~\ref{prop} guarantees that the difference of two possible solutions $\f d$ and $\dd$ of Eq.~\eqref{simple} measured by the relative energy~\eqref{relenergy} depends continuously on the difference of their initial values. The growth of the relative energy is  controlled only by the difference of the initial values and one of the two solutions, \emph{i.e.}, $\dd $ and the prescribed velocity field $\f v$. 
 One consequence is the  weak-strong uniqueness property of the system, \emph{i.e.}, the solutions must coincide if one is sufficiently regular (compare with Eq.~\eqref{K}) and both  emanate from the same initial data.

\section{Open problems}
\label{sect-open}

In view of the substantial research progress discussed in the previous sections, it seems
appropriate to close this article by pointing out some open problems and challenging questions
 for the future.
 
 We start by taking the mathematical perspective. Whereas the stationary problem seems under control
 (as summarized, \emph{e.g.}, by Ball~\cite{Ball2}), the dynamical behavior has many aspects
 for which a detailed mathematical
 treatment is still missing 
 (see, \emph{e.g.},~Zhou \emph{et\,al}.~\cite{zhou}). 
 All dynamical models presented in this article consist of an evolution
  equation for the quantity describing the anisotropy (distribution function, ${Q}$-tensor, director), which is then
  coupled to the Navier--Stokes equations
   (see~Eqs.\eqref{instatdoi}-\eqref{navdoi}, \eqref{evo} and~\eqref{ELeq}).
Showing the well-posedness for the Navier--Stokes equations alone is still a Millennium problem~(see
 Fefferman~\cite{mill}) that seems to be out of reach (see Tao~\cite{tao}). 
Consequently, one cannot hope to show well-posedness of the coupled
 systems discussed in this article. However, as already indicated in the previous sections, there are recent mathematical achievements from which future investigations
 could start.
 
 Specifically, in case of the Doi--Hess model~\eqref{instatdoi}-\eqref{navdoi}, so far only local existence of weak solutions has been shown (see~Zhang and Zhang~\cite{existinstatdoi}). 
 Here, global existence of solutions would be helpful
for implementing a suitable numerical (finite-element) scheme. 
For the $Q$-tensor model and the Ericksen--Leslie model, there are several results on global existence of generalized solutions. However, these have been proved for the special case of a free energy functional
involving the one-constant
 approximation (see, \emph{e.g.}, Paicu and Zarnescu~\cite{existence2} for the $Q$-tensor
  model~\eqref{evo} combined with the Landau--de Gennes free energy~\eqref{landau}, Wilkinson~\cite{mark}
   for the $Q$-tensor model~\eqref{evo} combined with the Ball--Majumdar free energy~\eqref{BM} and
    Cavaterra, Rocca and Wu for the Ericksen--Leslie model~\cite{allgemein}). 
    Recently, Huang and Ding~\cite{HuangDing} have suggested a partial generalization for the $Q$-tensor model~\eqref{evo},
      and a similar generalization has been carried out in Emmrich and Lasarzik~\cite{unsere} in case of the Ericksen--Leslie equation.  
      It would be very interesting to extend these analytical results towards more general free energies away from the one-constant approximation such as the full Landau--de Gennes expression~\eqref{landau} for the $Q$-tensor model or the full Oseen expression Eq.~\eqref{oseen} for the Ericksen--Leslie model  (see~Lasarzik~\cite{unseremasse}). 
Again, for these models global existence of generalized solutions would be 
desirable.

Progress in these directions is of major importance also from the perspective of soft condensed matter physics.
Indeed, as already mentioned in the introduction, there are many studies
where equations of the type discussed in this article, and even more complicated variants, are numerically solved without (mathematically confirmed) knowledge about solutions. 
For example, ${Q}$-tensor theories are nowadays often used to model {\em hybrid} systems involving additional degrees of freedom, such as
liquid crystals with embedded colloids (see Ref.~\cite{ravnik}), polar liquid crystals (see Ref.~\cite{Ilg2}), branched polymers (see Ref.~\cite{Ilg3}),

and ferrogels (see Ref.~\cite{menzel}).
%
Moreover, combining ${Q}$-tensor models
with activity terms (see Refs.~\cite{simha,heidenreich})
 has led to a boost of studies targeting the collective behavior of active systems such as bacterial or artificial 
microswimmers (see, \emph{e.g.}, Refs.~\cite{hemingway,heidenreich2}).
To our knowledge, these extensions of the theory have not been considered by mathematicians at all.

 Finally, we briefly comment on future applications of the relative energy approach introduced in Sec.~\ref{sect-rela}.
The presented result (Prop.~\ref{prop}) is a first step to generalize the concept of relative energy inequalities to systems with nonconvex energy functionals.
This technique and associated results can hopefully be transferred to other systems with similar properties. 
A result implicated by the relative energy inequality is the so-called weak-strong uniqueness: 
In the case that a solution admitting additional regularity exists, it coincides with all other possible generalized solutions emanating form the same initial data.
 Indeed, it can be immediately inferred from Proposition~\ref{prop} that the relative energy vanishes if the initial values of the two solutions $\f d $ and $\dd$ coincide and $\dd$ fulfills the additional regularity requirements in Eq.~\eqref{K}.
  Thus, the solutions $\f d$ and $\dd$ coincide as long as the more regular solution $\dd$ exists.  
  The relative energy inequality~\eqref{relEn} can also be used to derive a weakened solution formalism, the so-called dissipative solution, where only an inequality is assumed to be fulfilled by the solution (see Ref.~\cite{lionsfluid}). Beside these mathematical questions, it is also possible to derive \emph{a posteriori} estimates to bound modeling errors as well as errors due to numerical approximation (see Ref.~\cite{fischer}).

\section*{Acknowledgments}

We gratefully acknowledge financial support from the Deutsche Forschungsgemeinschaft through
the Collaborative Research Center~901 "Control of self-organizing nonlinear systems: Theoretical methods and concepts of application" (projects A8 and B2).


\end{document}